\documentclass[aps,english,prl,floatfix,amsmath,superscriptaddress,tightenlines,twocolumn,nofootinbib]{revtex4-2}
\usepackage{mathtools, amssymb}
\usepackage{graphicx}
\usepackage{tikz}
\usepackage{tikz-3dplot}
\usetikzlibrary{spy}
\usetikzlibrary{arrows.meta}
\usetikzlibrary{calc}
\usetikzlibrary{decorations.pathreplacing,calligraphy}
\usepackage{intcalc}
\usepackage[caption=false]{subfig}
\usepackage[shortlabels]{enumitem}
\usepackage{soul}
\usepackage[utf8]{inputenc}
\usepackage{xcolor}
\usepackage{amsthm}
\usepackage{float}
\usepackage{comment}

\renewcommand{\bar}{\overline}

\newcommand*{\Scale}[2][4]{\scalebox{#1}{$#2$}}

\makeatletter 

\renewcommand\onecolumngrid{% <<<<<<
	\do@columngrid{one}{\@ne}%
	\def\set@footnotewidth{\onecolumngrid}% <<<<<<<<<<<<<<<<
	\def\footnoterule{\kern-6pt\hrule width 1.5in\kern6pt}%
}

\renewcommand\twocolumngrid{% <<<<<<
	\def\footnoterule{% restore rule
		\dimen@\skip\footins\divide\dimen@\thr@@
		\kern-\dimen@\hrule width.5in\kern\dimen@}
	\do@columngrid{mlt}{\tw@}
}%

\makeatother    
\newcommand{\Tr}[0]{{\text{Tr}}}
\theoremstyle{plain}
\newtheorem{theorem}{Theorem}%[section] % reset theorem numbering for each chapter

\theoremstyle{definition}
\newtheorem{exmp}[theorem]{Example} % same for example numbers
\newtheorem*{remark}{Remark}
\theoremstyle{plain}

\theoremstyle{plain} 
\newtheorem{lemma}[theorem]{Lemma}

\theoremstyle{plain}

\theoremstyle{plain} 
\newtheorem{Proposition}[theorem]{Proposition}

\usepackage{hyperref}
\hypersetup{colorlinks=true,linkcolor=red,citecolor=magenta,urlcolor=blue}
\usepackage[all]{hypcap}

\begin{document}

	\title{Modular Commutators in Conformal Field Theory}
	
	\author{Yijian Zou}
	\thanks{These authors contributed equally.}
	\affiliation{Stanford Institute for Theoretical Physics, Stanford University, Palo Alto, CA 94305, USA}
    \author{Bowen Shi}
    \thanks{These authors contributed equally.}
    \affiliation{Department of Physics, University of California at San Diego, La Jolla, CA 92093, USA}
    \author{Jonathan Sorce}
	\affiliation{Stanford Institute for Theoretical Physics, Stanford University, Palo Alto, CA 94305, USA}
	\author{Ian T. Lim}
	\affiliation{Department of Physics and Astronomy, University of California, Davis, CA 95616, USA}
	\author{Isaac H. Kim}
	\affiliation{Department of Computer Science, University of California, Davis, CA 95616, USA}
	\date{\today}
	
	\begin{abstract} 
		The modular commutator is a recently discovered multipartite entanglement measure that quantifies the chirality of the underlying many-body quantum state. In this Letter, we derive a universal expression for the modular commutator in conformal field theories in $1+1$ dimensions and discuss its salient features. We show that the modular commutator depends only on the chiral central charge and the conformal cross ratio. We test this formula for a gapped $(2+1)$-dimensional system with a chiral edge, i.e., the quantum Hall state, and observe excellent agreement with numerical simulations. Furthermore, we propose a geometric dual for the modular commutator in certain preferred states of the AdS/CFT correspondence. For these states, we argue that the modular commutator can be obtained from a set of crossing angles between intersecting Ryu-Takayanagi surfaces.
	\end{abstract}
	\maketitle
	
	One of the overarching themes of research in theoretical physics over the past few decades has been the study of entanglement in interacting quantum many-body systems. Calculation of the canonical measure of entanglement --- entanglement entropy --- has played a crucial role in elucidating the physics of topological order~\cite{Kitaev2006,Levin2006}, conformal field theory~\cite{Calabrese2004}, and holographic duality~\cite{Ryu2006}. 
	
	Recently, a new entanglement measure known as the \emph{modular commutator}  was introduced~\cite{Kim2021,Kim2021a}. The modular commutator is defined as $J(A,B,C)_{\rho} := i\text{Tr}(\rho_{ABC}[\ln \rho_{AB}, \ln \rho_{BC}])$ for a tripartite quantum state $\rho_{ABC}$, and unlike other known entanglement measures, it is odd under time reversal. In the context of topologically ordered systems in $2+1$D, the modular commutator was used to extract the chiral central charge of the edge theory~\cite{Kim2021,Kim2021a}.

In this Letter, we derive a universal expression for the modular commutator in conformal field theories in $1+1$D and discuss its physical implications. Let $A, B,$ and $C$ be three contiguous spacetime intervals; see Fig.~\ref{fig:setup}(a). In this setup, we derive a general expression for $J(A,B,C)$ in the vacuum. If the subsystems lie in a single time-slice, the expression simplifies to
	\begin{equation}\label{eq:main_result}
		J(A,B,C)_{|\Omega\rangle}= \frac{\pi c_-}{6}(2\eta - 1), 		
	\end{equation}
where $\eta = \frac{(x_2-x_1) (x_4-x_3)}{(x_3 - x_1)(x_4-x_2)}$ is the cross ratio, $c_- = c_L - c_R$ is the chiral central charge of the CFT, and $|\Omega \rangle$ is the vacuum state. Using a standard conformal mapping from the complex plane to the cylinder,
expressions for the modular commutator for finite systems in the vacuum and infinite systems at finite temperature are also derived. 
	
We primarily discuss two applications. First, we argue that Eq.~\eqref{eq:main_result} can be a useful tool to study the entanglement structure of 2+1D chiral gapped systems at their edges.  Specifically, consider three contiguous intervals $A, B,$ and $C$ at the edge of a disk; see Fig.~\ref{fig:edge}(a). We propose the following formula --- based on an argument utilizing Eq.~\eqref{eq:main_result} --- for the modular commutator:
	\begin{equation}
		J(A,B,C)_{|\psi_{2D}\rangle} = \frac{\pi \mathfrak{c_-}}{3}\eta, \label{eq:jabc_edge}
	\end{equation}
where $\mathfrak{c_-}$ is the chiral central charge of the 2+1D system (defined as a coefficient appearing in the edge energy current~\cite{Kane1997,Read2000,Kitaev2006a}) and $| \psi_{\rm 2D}\rangle$ is the ground state. We test Eq.~\eqref{eq:jabc_edge} numerically for the Chern insulator and $p+ip$ topological superconductor, demonstrating excellent agreement.

When $A, B$ and $C$ cover the entire edge (see Fig.~\ref{fig:edge}(b)), i.e., $\eta=1$, we provide an independent information-theoretic argument for a stronger result:
	\begin{equation}
		J(A,B,C)_{|\widetilde{\psi}_{2D}\rangle} = \frac{\pi}{3} \mathfrak{c_-},  \label{eq:jabc_edge_cover}
	\end{equation}
where  $|\widetilde{\psi}_{2D}\rangle$ is any state which is indistinguishable from $|\psi_{2D}\rangle$ in the bulk region. We emphasize the generality of Eq.~\eqref{eq:jabc_edge_cover} in two directions. First, this equation holds even if there is an excitation localized at the edge. Second, the argument continues to hold even if the shape of the edge is deformed continuously. The underlying argument --- based on the properties of modular commutator~\cite{Kim2021,Kim2021a} and techniques from the entanglement bootstrap~\cite{Shi2020} --- reveals that the robustness of this result originates from the entanglement area law of the bulk. 

Second, we propose a holographic interpretation of Eq.~\eqref{eq:main_result}. Our interpretation rests on an observation that Eq.~\eqref{eq:main_result} can be recast as 
	\begin{equation}
		J(A,B,C)_{|\Omega\rangle} = \frac{\pi c_-}{6} \cos \theta,
		\label{eq:Jabc_holography}
	\end{equation}
where $\theta$ is the crossing angle of the two geodesics (i.e., two Ryu-Takayanagi surfaces~\cite{Ryu2006}) in $\mathrm{AdS}_3$, each anchored at the boundaries of $AB$ and $BC$, respectively. We verify this correspondence at both zero and finite temperature and propose a generalization to any state whose bulk geometry has a moment of time symmetry~\cite{engelhardt2014extremal}.\footnote{A spacelike slice is a \textit{moment of time symmetry} if its extrinsic curvature tensor vanishes; this is like an infinitesimal version of time-reflection symmetry. In this case, Ryu-Takayanagi surfaces of boundary regions on the slice lie entirely within the slice.}

	\begin{figure}[h]
		\centering
		\includegraphics[scale=0.93]{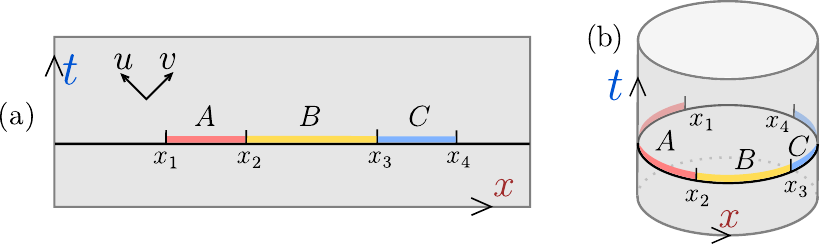}
		\caption{(a) Three contiguous intervals $A, B,$ and $C$, on a single-time slice.
			(b)  Contiguous intervals on a circle $S^1$ with circumference $L$.  }
		\label{fig:setup}
	\end{figure}

Our approach to derive Eq.~\eqref{eq:main_result} will be geometric in nature. The main advantage of this derivation is that it makes the generalization of Eq.~\eqref{eq:main_result} to arbitrary spacetime intervals straightforward. Alternative derivations shall be discussed in the Supplementary Material (SM) as well.
	
\emph{Geometric derivation---} Our derivation of Eq.~\eqref{eq:main_result} is based on the following two observations. First, the modular commutator $J(A,B,C)$ can be viewed as the linear response of the $BC$ entanglement entropy under the $AB$ modular flow~\cite{Kim2021a,chen2018modular,faulkner2019modular}. Second, for a 1+1D CFT, the modular flow for a finite interval generates a special conformal transformation that keeps the two ends of the interval fixed~\cite{Casini2011,BW1975,BW1976}. Thus, we will compute the modular commutator $J(A,B,C)$ by the change of the entropy $S_{BC}$ from the infinitesimal conformal transformation generated by the modular flow corresponding to $AB$.
	
The \emph{modular flow} of an operator $O$ with respect to a state $\rho$ and a subsystem $A$ is defined as $O(s) := \rho_A^{is} O \rho_A^{-is}$ for $s \in \mathbb{R}$. We consider the action of the modular flow associated with the interval $AB$ in the vacuum. Define the following one-parameter family of density matrices: $\rho_{ABC}(s) := \rho^{is}_{AB} \rho_{ABC} \rho^{-is}_{AB}$. The response of the von Neumann entropy of $\rho_{BC}(s) = \Tr_{A}(\rho_{ABC}(s))$ under this flow is related to the modular commutator by~\cite{Kim2021a}: 
	\begin{equation}
		\left.\frac{dS(\rho_{BC}(s))}{ds}\right|_{s=0}= -J(A,B,C)_{\rho},
	\end{equation}
with $S(\rho) := -\Tr(\rho\ln\rho).$

In quantum field theory, the observables restricted to the interval $AB$ completely determine the observables in the full \textit{causal diamond} $D(AB),$ i.e., the domain of dependence of $AB$. In 1+1D CFT, the modular flow associated to a spacelike interval in the vacuum is a local transformation of observables lying within its causal diamond \cite{Casini2011}. The relevant vector fields are illustrated in Fig.~\ref{fig:causal-diamond}. 
	\begin{figure}[h]
	\centering
	
	\begin{tikzpicture}
	\node [] at (0,0) {\includegraphics[scale=1.21]{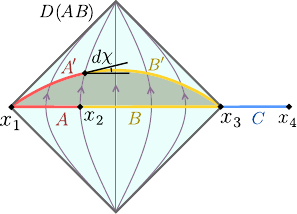}};
	
	\begin{scope}[xshift=4.35 cm,yshift=0.0 cm]
		\draw [fill=purple!60!blue!15!white] (-2.2,-0.86) rectangle (2.2,0.97);
	\end{scope}
	\begin{scope}[yshift=-0.17 cm]
		\begin{scope}[xshift=4.35 cm,yshift=0.85 cm]
			\node[]  at (0,0) {\footnotesize{Within the causal diamond:}};
		\end{scope}	
		
		\begin{scope}[xshift=4.35 cm,yshift=0.0 cm]
			\node[]  at (0,0) {$\Scale[0.98]{\left(\begin{array}{l}
					\Scale[1.03]{\frac{du}{ds}}\\[3.7pt]
					\Scale[1.03]{\frac{dv}{ds}}
				\end{array}
				\right)\Scale[0.9]{=2\pi} \left(\begin{array}{l}
					\Scale[1.03]{\frac{(u-u_1)(u_3-u)}{u_{1}-u_{3}}}\\[3.7pt]
					\Scale[1.03]{\frac{(v-v_1)(v_3-v)}{v_{3}-v_{1}}}
				\end{array}\right)}$};
		\end{scope}
	\end{scope}
	\end{tikzpicture}
	\caption{Modular flow in the interior of the causal diamond $D(AB)$ and the associated vector field. Under an infinitesimal flow by a parameter $ds$, interval $AB$ becomes $A'B'$ and a boost angle $d\chi$ develops at the left end of $B'$. }
	\label{fig:causal-diamond}
\end{figure}
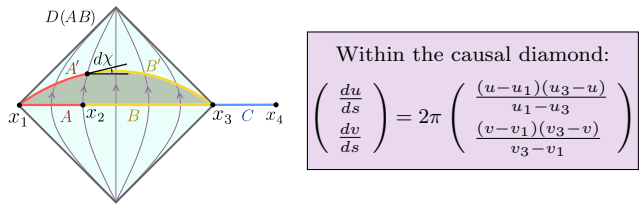

Now we can use the following regulated form of the single-interval entanglement entropy for chiral CFTs in 1+1D~\cite{Iqbal2016}:
	\begin{equation}
		\label{eq:SBC}
		S_{BC} = \frac{c_L}{12}\ln\frac{(v_4-v_2)^2}{\epsilon_{v2}\epsilon_{v4}} + \frac{c_R}{12}\ln\frac{(u_4-u_2)^2}{\epsilon_{u2}\epsilon_{u4}},
	\end{equation}
where $u=t-x$ and $v=t+x$ are light-cone coordinates, and $\epsilon_{u(v)2(4)}$ denotes the UV cutoffs in the $u$ and $v$ directions at the endpoints $x_2$ and $x_4$.
	
Note that the point $x_4$ is unaffected by the modular flow with respect to $AB$, because it is outside $D(AB)$. Thus, $u_4,v_4$ and $\epsilon_{u4},\epsilon_{v4}$ remain unchanged; the change only occurs at $x_2$. Importantly, the cutoffs $\epsilon_{u(v)2(4)}$ transform nontrivially under local diffeomorphisms. They are rescaled by the local boost angle (see Fig.~\ref{fig:causal-diamond}),
	\begin{equation}\label{eq:diff-cutoffs}
		d \ln \epsilon_{v2} = -d\ln \epsilon_{u2} = d \chi, 
	\end{equation}
where $d \chi = \frac{2\pi (x_{23}-x_{12})}{x_{13}} ds$ is the boost angle at $x_2$. Here we use the convention $x_{ij} = x_j - x_i$. Differentiating Eq.~\eqref{eq:SBC} and using Eq.~\eqref{eq:diff-cutoffs} we obtain
	\begin{equation}
		\label{eq:mainJ}
		J(A,B,C)_{\vert \Omega \rangle}=  \frac{\pi c_{-}}{6}(2\eta-1),
	\end{equation}
where the chiral central charge is $c_{-}=c_L-c_R$ and the cross ratio is $\eta = \frac{x_{12}x_{34}}{x_{13}x_{24}}$.  Generalization of Eq.~\eqref{eq:mainJ} to general Cauchy surfaces is straightforward, and can be used to determine $c_L$ and $c_R$ individually in terms of the modular commutator; see the SM for details.
	
Eq.~\eqref{eq:mainJ} for $J(A,B,C)_{|\Omega\rangle}$ possesses a set of important properties, summarized below. Firstly, $J$ is odd under time reversal, which exchanges $c_L$ and $c_R$. This is in contrast with other entanglement measures such as the entanglement entropy, which are even under time reversal. Secondly, $J$ is odd under the map $\eta \rightarrow 1-\eta$. In particular, $J=0$ at $\eta = 1/2$ where the modular commutator changes sign. Thirdly, as the length of one interval gets small, $J$ does not vanish but takes on universal values. As $x_1\rightarrow x_2$ or $x_3\rightarrow x_4$, $\eta\rightarrow 0$ and $J\rightarrow -\pi c_{-}/6$, and similarly, as $x_2\rightarrow x_3$, $\eta\rightarrow 1$ and $J\rightarrow \pi c_{-}/6$. In fact, we shall later see that the universal difference $J(\eta=1)-J(\eta=0) = \pi c_{-}/3$ is exactly the modular commutator for 2D chiral topological order. Lastly, if $c_-\ne 0$, we have $J=\pi c_{-}/6 \ne 0$ when $ABC$ is the entire circle. This distinguishes $|\Omega\rangle$ from any pure state on a Hilbert space factorized into a tensor product on spatial regions, as the latter necessarily has $J=0$. Thus, $c_{-}\neq 0$ is incompatible with any lattice regularization (see also \cite{Hellerman2021} for an alternative argument). 

More generally, one can consider a thermal state at inverse temperature $\beta$ on a circle of circumference $L$, denoted as $\rho^{(\beta;L)}$. Through standard conformal mappings from planes to cylinders \cite{cardy_entanglement_2016}, one can show that the modular commutator $J(A,B,C)$ remains to be in the form in Eq.~(\ref{eq:main_result}) in two limits $\beta/L \rightarrow 0, \infty,$ with
	%given that 
	the cross ratio $\eta$ replaced by $\eta_{\rm eff}^{(\beta;L)}$:
	\begin{equation}\label{eq:eta-eff}
		\eta_{\rm eff}^{(\beta; L)}= \left\{ \begin{array}{ll}
			\Scale[1.22]{\frac{\sin(\pi x_{12}/L)\sin(\pi x_{34}/L)}{\sin(\pi x_{13}/L)\sin(\pi x_{24}/L)}}, & \Scale[0.9]{\beta/L \rightarrow \infty,} \\[6.7pt]
			\Scale[1.2]{\frac{\sinh(\pi x_{12}/\beta)\sinh(\pi x_{34}/\beta)}{\sinh(\pi x_{13}/\beta)\sinh(\pi x_{24}/\beta)}}, & \Scale[0.9]{L/\beta \rightarrow \infty}.
		\end{array}\right.
	\end{equation}

\emph{Chiral thermal state---} The modular commutator can be nonzero even for non-chiral CFTs, provided that the temperatures for the left- and the right-moving modes are unequal. We refer to such states as \emph{chiral thermal states} \cite{Bernard2012,Tu2013,Bhaseen2015}: 
	\begin{equation}
		\rho^{(\beta_L,\beta_R;L)}  = \frac{1}{Z} e^{-\beta_L H_L -\beta_R H_R}.
	\end{equation}
Here $H_L$ and $H_R$ are the Hamiltonians of the left- and right-moving sectors, respectively. Similarly, $(\beta_L, \beta_R)$ represent inverse temperatures for the respective modes.

There are a few reasons to study chiral thermal states. First, a chiral thermal state can be obtained by applying the Lorentzian boost to a thermal state. Second, there are concrete lattice models whose underlying state at low temperature can be well-described by a chiral thermal state. For instance, it was noted that the reduced density matrix near the edge of a chiral topological order in $2+1$D can be represented by a chiral thermal state with $(\beta_L,\beta_R)=(\infty,{\textrm{finite}})$~\cite{Tu2013}. Third, as we show in the SM, one can sometimes explicitly construct chiral thermal states in lattice models, making the numerical verification tractable.

From Eq.~\eqref{eq:eta-eff}, for a general chiral thermal state $\rho^{(\beta_L,\beta_R;L)}$ we have
	\begin{equation}
		\label{eq:J_chiral_thermal}
		J(A,B,C)_{\rho^{(\beta_L,\beta_R;L)}} = \frac{\pi}{3}c( \eta_{\rm eff}^{(\beta_L;L)} - \eta_{\rm eff}^{(\beta_R;L)}),
	\end{equation}
where $c=c_R=c_L$. We test Eq.~\eqref{eq:J_chiral_thermal} for the lattice chiral thermal states and find excellent agreement (see SM for details). 

\emph{Edge of 2+1D chiral topological order ---}
The chiral thermal state can provide insights into the edges of 2+1D gapped systems with non-zero chiral central charge, denoted as $\mathfrak{c_-}$~\cite{Kane1997,Read2000,Kitaev2006a,Tu2013}. (We choose a different font to distinguish two concepts: the chiral central charge $\mathfrak{c_-}$ of a 2+1D gapped phase versus $c_-$ for a 1+1D chiral CFT.) 

\begin{figure}[h]
	\centering
	\includegraphics[width=0.94\linewidth]{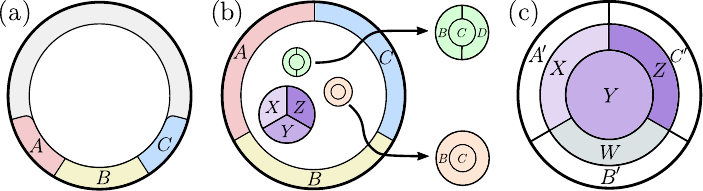}
	\caption{A 2+1D gapped chiral system on a disk and various choices of subsystems. The sizes (width) for subsystems within the bulk (adjacent to the edge) are large compared to the bulk correlation length. }\label{fig:edge}
\end{figure}

Consider a ground state $ |\psi_{\rm 2D} \rangle$ on a disk for concreteness; see Fig.~\ref{fig:edge}. For an annulus which covers the entire edge, e.g., the annulus in Fig.~\ref{fig:edge}(a), the reduced density matrix of  $|\psi_{\rm 2D} \rangle$, can be viewed as a 1+1D system. If the edge is completely chiral (that is when, e.g., it only has left-moving modes but not right-moving modes), it is expected to be described by a chiral thermal state whose $c$ equals $\mathfrak{c_-}$~\cite{Tu2013}.  

Then by applying Eq.~\eqref{eq:J_chiral_thermal}
to the interval choice in Fig.~\ref{fig:edge}(a) and taking $\beta_L=\infty$, $\beta_R \ll L_A,L_B,L_C$ (the lengths of the regions), 
we arrive at a prediction 
\begin{equation}\label{eq:eta-2D}
	J(A,B,C)_{\vert \psi_{\rm 2D}\rangle } = \frac{\pi}{3} \mathfrak{c_-} \eta.
\end{equation}
We have tested this formula numerically for a Chern insulator and observed excellent agreement; see Fig.~\ref{fig:Chern-insulator-2}. We propose this formula to hold for general translation invariant topologically ordered systems in 2+1D. 

\begin{figure}
		\centering
		\begin{tikzpicture}
			\node[] (A) at (0,0) {\includegraphics[width=0.48\linewidth]{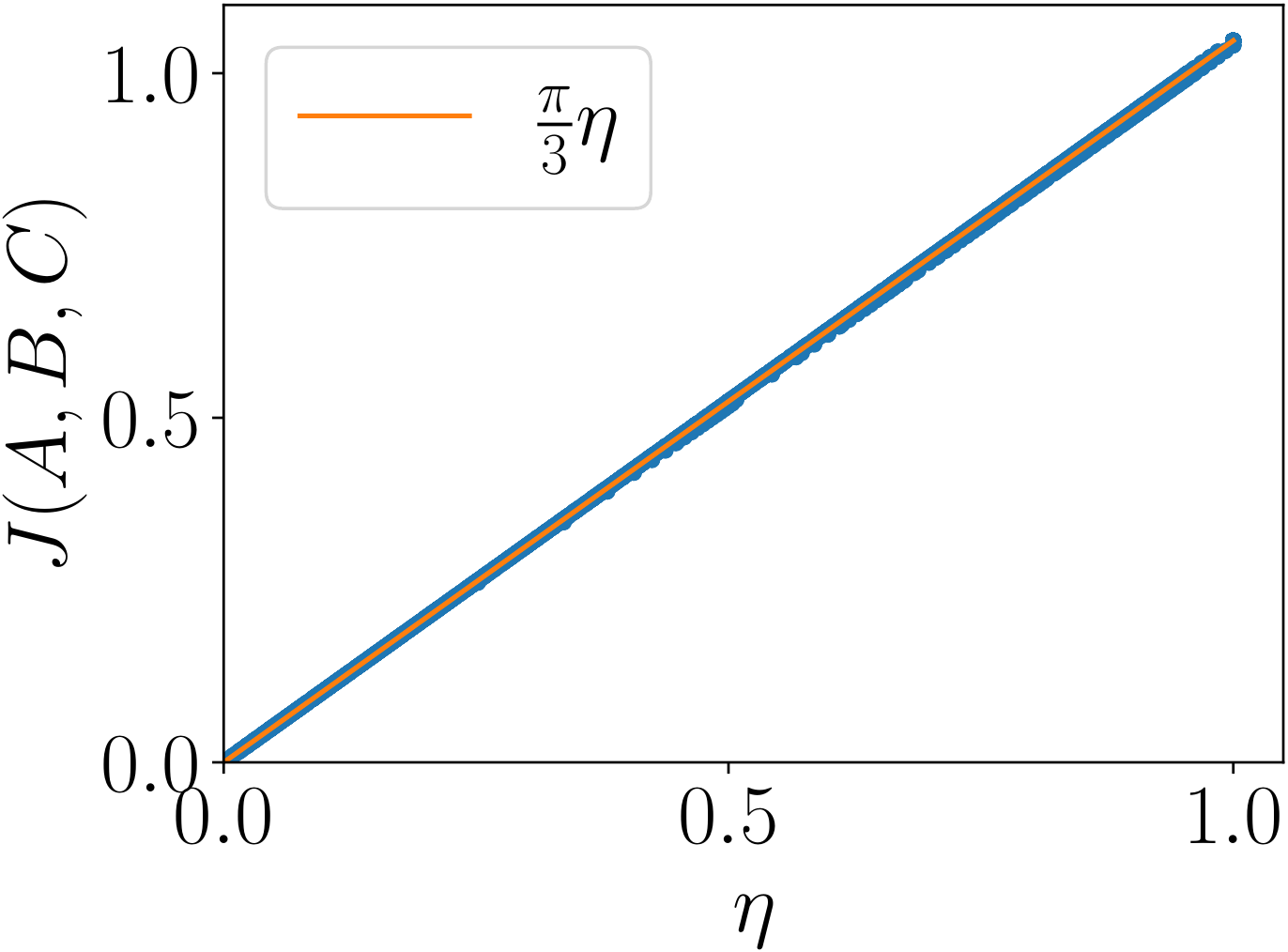}};
			\node[] (B) at (4.3,0) {\includegraphics[width=0.48\linewidth]{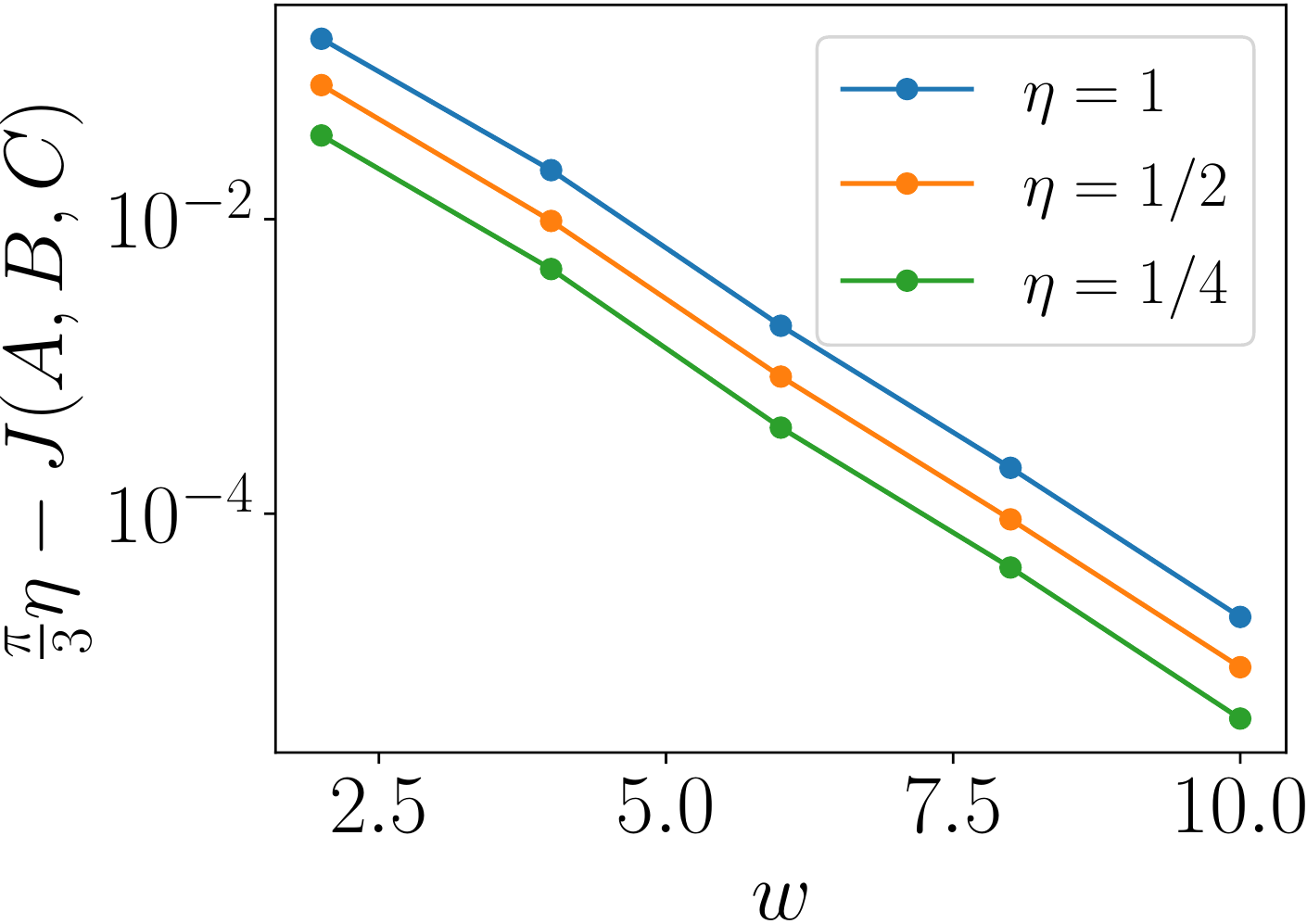}};
		\end{tikzpicture}
		
		\caption{$J(A,B,C)$ versus $\eta$ for the Chern insulator, which is realized by filling the lowest band of the Hofstadter model with flux $\pi/2$. We use a square lattice on a cylinder with circumference $L=144$ and height $W=32$. $A,B,C$ are rectangular strips on the boundary with length $L_A,L_B,L_C$ and width $w$. Left: We fix $w=10$ and vary the lengths $L_A,L_B,L_C$. Right: We choose several $(L_A,L_B,L_C)$ and vary $w$. The three choices $(L_A,L_B,L_C) = (48,48,48), (36,36,36), (24,48,24)$ correspond to $\eta=1,1/2,1/4$, respectively.}
	\label{fig:Chern-insulator-2}
\end{figure}

\emph{Topological argument---} 
When the union of intervals $A$, $B$, and $C$ is the entire annulus, as shown in Fig.~\ref{fig:edge}(b), Eq.~\eqref{eq:eta-2D} becomes $J=\frac{\pi }{3} \mathfrak{c_-}$.
Here we present an entirely different argument for this formula, based on the entanglement area law of the 2+1D bulk~\cite{Kitaev2006,Levin2006}. Our argument reveals an extra degree of robustness of this expression:
\begin{equation}\label{eq:eta-2D-entire}
	J(A,B,C)_{\vert \widetilde{\psi}_{\rm 2D}\rangle } = \frac{\pi}{3} \mathfrak{c_-} \quad \textrm{for Fig.~\ref{fig:edge}(b)}.
\end{equation}
We show that Eq.~\eqref{eq:eta-2D-entire} holds for any state $|\widetilde{\psi}_{\rm 2D}\rangle$ locally indistinguishable from the ground state within the bulk. Note that we need not assume $|\widetilde{\psi}_{\rm 2D}\rangle$ to be the ground state; our argument applies even if there are edge excitations, as long as the global state is pure. 

The key observation that leads to Eq.~\eqref{eq:eta-2D-entire} is an equivalence we will establish between the edge and the bulk modular commutator for the set of subsystems shown in  Fig.~\ref{fig:edge}(b):
\begin{equation}\label{eq:2Js=}
	J(A,B,C)_{|\widetilde{\psi}_{\rm 2D}\rangle}=-J(X,Y,Z)_{|\widetilde{\psi}_{\rm 2D}\rangle }.
\end{equation}
Note that the regions $A, B,$ and $C$ lie at the edge while the regions $X, Y,$ and $Z$ lie entirely in the bulk. 
%Thus, $J(A,B,C)_{|\widetilde{\psi}_{\rm 2D}\rangle}$ is the modular commutator at the edge and $J(X,Y,Z)_{|\widetilde{\psi}_{\rm 2D}\rangle}$ is the modular commutator in the bulk. 
Once this relation is established, one can use the formula for the bulk modular commutator~\cite{Kim2021}, i.e., $J(X,Y,Z)_{|\widetilde{\psi}_{\rm 2D}\rangle} = -\frac{\pi}{3}\mathfrak{c_-}$ to complete the derivation. 

The equivalence of the two modular commutators directly follows from Section VI of Ref.~\cite{Kim2021a}, as we explain below. (See SM for a more detailed explanation.) First of all, the state $| \widetilde{\psi}_{\rm 2D}\rangle$, being indistinguishable from the ground state in the bulk, satisfies the axioms of entanglement bootstrap~\cite{Shi2020}. Of particular importance to us is the axiom \textbf{A1} in Ref.~\cite{Shi2020}, which holds for local disk-like regions away from the edge; it says $(S_{BC}+S_{CD}-S_B-S_D)_{| \psi_{\rm 2D}\rangle}=0$ for the green disk $BCD$ shown in Fig.~\ref{fig:edge}(b), where $|\psi_{\rm 2D}\rangle$ is the ground state. This axiom, applied to the bulk disk $XYZW$ of Fig.~\ref{fig:edge}(c), gives $I(A':Y|X)=I(C':Y|Z)=0$. It then follows that, for state $|\widetilde{\psi}_{\rm 2D}\rangle$:
\begin{equation}
	J(X,Y,Z)= J(A'X,Y,C'Z) =-J(A'X,WB',C'Z). \nonumber
\end{equation}
Letting $A=A'X$, $B=B'W$, and $C=C'Z$, we conclude Eq.~\eqref{eq:2Js=}. 

Let us emphasize the generality of the argument above. Note that nowhere in the derivation did we use any symmetry (e.g., translation or rotation symmetry) nor did we use any condition of the state in the vicinity of the edge. For instance, even in the presence of strong disorder, even though the conformal symmetry does not hold --- not even approximately --- formula \eqref{eq:eta-2D-entire} still holds; this is numerically verified for integer quantum Hall states, see the SM. Moreover, the argument holds as long as $\vert \widetilde{\psi}_{\rm 2D}\rangle=U_{\rm edge} \vert \psi_{2D}\rangle$, where $U_{\rm edge}$ is \emph{any} unitary operator along the edge which is thin compared to the width of the subsystems; specifically, $U_{\rm edge}$ should be supported within the annulus $A'B'C'$ for the choice of $ABC$ in Fig.~\ref{fig:edge}(c). (Under a plausible assumption, the unitarity assumption can be dropped. See SM for the detail.)

\emph{Holographic interpretation---} 
In the AdS/CFT correspondence~\cite{Maldacena1999}, entanglement quantities of the boundary CFT are mapped to geometric quantities in the bulk of an asymptotic AdS space. For example, the Ryu-Takayanagi (RT) formula~\cite{Ryu2006} implies that in ordinary non-chiral AdS/CFT, the entanglement entropy of a  boundary region $A$ in a time-symmetric state is given by the minimal length of the bulk geodesic $\gamma_A$ (also known as the RT surface) homologous to the region. Some examples are shown in Fig.~\ref{fig:AdS}.

\begin{figure}[h]
	\centering
	\includegraphics[width=0.98\linewidth]{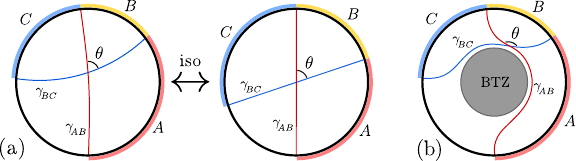}
	\caption{Verified cases of the holographic conjecture: (a) At zero temperature. Each disk is a Poincar\'e disk and the two are related by an isometry. (b) At a finite (high) temperature such that $\beta \ll L$. }
	\label{fig:AdS}
\end{figure}

Here we propose to extend the holographic dictionary to the modular commutator for chiral realizations of AdS/CFT, e.g. \cite{Li2008}. In states whose bulk geometries are locally AdS$_3$\footnote{The restriction of the conjecture to locally AdS$_3$ spacetimes is because, in chiral gravity, these are the spacetimes where bulk geodesics are used to compute boundary modular Hamiltonians~\cite{castro2014holographic}.} with a moment of time symmetry, we propose: 
\begin{equation}\label{eq:general_J_AdS}
	J(A,B,C) = \frac{ \pi c_-}{6} \sum_i \cos\theta_i,
\end{equation}
where $\{\theta_i \}$ is the set of crossing angles of the RT surfaces, i.e., geodesics $\gamma_{AB}$ and $\gamma_{BC}.$ Each $\theta_i$ is chosen such that $\gamma_{AB}$, seen inwardly, lies at the right side of the angle; see Fig.~\ref{fig:AdS} for examples. In general, $AB$ and $BC$ may have multiple connected components; see SM for the relevant discussion.

We can verify the conjecture for a few simple cases shown in Fig.~\ref{fig:AdS}. The vacuum state of chiral AdS$_3$/CFT$_2$ is described by the ordinary vacuum AdS$_3$ spacetime~\cite{Li2008}. On the $t=0$ slice of this spacetime, we can apply a bulk isometry to place the intersection point of any two geodesics at the center of the Poincar\'e  disk. Then the two geodesics become straight lines with a crossing angle $\theta$. Since the cross ratio $\eta$ --- given by $\eta = (x_{12}x_{34})/(x_{13}x_{24})$ --- is preserved under this isometry, the identity $2\eta-1 = \cos\theta$ follows from simple trigonometry. Thus, we arrive at
\begin{equation}
	J(A,B,C)_{\vert \Omega\rangle} = \frac{\pi c_-}{6} \cos \theta.
\end{equation}
At high temperatures $\beta \ll L$, thermal states in CFT are dual to BTZ black holes~\cite{BTZ1992} in the bulk; see Fig.~\ref{fig:AdS}(b). An analogous derivation applies because the BTZ black hole can be viewed as a quotient of global $\mathrm{AdS}_3$. The result confirms our conjecture. (See SM for details.) 

In the semiclassical limit of AdS/CFT, a boundary modular Hamiltonian $K$ is dual to a bulk geometric operator which, in non-chiral AdS/CFT, is proportional to the area of the RT surface~\cite{jafferis2016gravity, Jafferis2016}. In chiral AdS/CFT, the operator has additional terms~\cite{castro2014holographic}; we will call the full operator $F.$
The modular commutator of contiguous intervals can be written in terms of commutators of $F$ operators. This commutator is zero in the vacuum for a single time-slice if the chiral central charge is zero \cite{Kaplan2022}, which matches Eq.~\eqref{eq:main_result}. However, for chiral theories, Eq.~\eqref{eq:general_J_AdS} implies the uncertainty relation
\begin{equation}\label{eq:F-main}
	\Delta F(AB) \cdot \Delta F(BC) \geq \frac{\pi c_- }{12}|\cos \theta|.
\end{equation}
Thus, the uncertainty in the geometric operator $F$ grows parametrically with the chiral central charge.

\emph{Discussion---} 
In this Letter, we computed the modular commutator~\cite{Kim2021,Kim2021a} in 1+1D CFTs, arriving at a simple formula Eq.~\eqref{eq:main_result} and discussing its applications in condensed matter systems and holography. For future work, it will be interesting to verify our conjecture in AdS/CFT to more general setups, e.g., disconnected intervals, states whose bulk geometries have no moment of time symmetry, and states with bulk quantum matter. Another interesting open problem is how our conjecture generalizes to higher dimensions. On the condensed matter side, it would be interesting to understand how Eqs.~\eqref{eq:eta-2D} and \eqref{eq:eta-2D-entire} generalize when the sector of the chiral edge is modified by an anyon in the bulk. 

\emph{Note added---} After posting this manuscript, we notice a related work \cite{fan2022entanglement} which has some overlap with this work. 

\emph{Acknowledgments:---} 
We thank Molly Kaplan, Nima Lashkari, Ruihua Fan, John McGreevy, Don Marolf, Mukund Rangamani, Yiming Chen and Dan Ranard for helpful discussions. B.S. is supported by the University of California Laboratory Fees Research Program, grant LFR-20-653926, and the Simons Collaboration on Ultra-Quantum Matter, grant 652264 from the Simons Foundation. Y.Z. is supported by the Q-FARM fellowship at Stanford University. I.L. is supported by a UC Davis Graduate Program Fellowship. J.S. is supported by AFOSR award FA9550-19-1-0369, CIFAR, DOE award DE-SC0019380 and the Simons Foundation.

\bibliography{bib}
\newpage

\onecolumngrid
\appendix

\section{Geometric derivation for general Cauchy surface}\label{sec:Cauchy}

We present a geometric derivation of a formula for the modular commutator for three contiguous intervals lying on a general Cauchy surface in the vacuum state of a 1+1D conformal field theory (CFT). In 1+1D, a Cauchy surface is a spacelike curve which generalizes the idea of a time-slice.\footnote{For an interested reader, here is a more precise definition. A Cauchy surface in a Lorentzian manifold $M$ is an achronal spacelike hypersurface $\Sigma$ with the property that if $p$ is a point in $M$ not in $\Sigma$, then every inextendible causal path through $p$ intersects $\Sigma$. See e.g. \cite{Witten2019} for details.}
See Fig.~\ref{fig:general-Cauchy} for an example. For any Cauchy surface, we can define a Hilbert space and a Schrodinger-picture representation of any global quantum state. The modular commutator can then be defined for subregions $A,$ $B,$ and $C$ of the Cauchy surface.

For any three contiguous segments $A,$ $B,$ $C$ on Cauchy surface $\Sigma$, our formula for the modular commutator reads:
\begin{equation}
\label{eq:1111}
\boxed{J(A,B,C)_{| \Omega_{\Sigma}\rangle} = \frac{\pi c_L}{6} (2\eta_v -1)-\frac{\pi c_R}{6} (2\eta_u -1),}
\end{equation}
where $\eta_u = \frac{u_{12}u_{34}}{u_{13}u_{24}},\eta_v = \frac{v_{12}v_{34}}{v_{13}v_{24}}$, and $|\Omega_{\Sigma}\rangle$ denotes the representation of the vacuum state on $\Sigma$. Here $u=t-x$ and $v=t+x$ are light-cone coordinates and $u_{ij} = u_j - u_i, v_{ij} = v_j- v_i$. The indices represent the end points of the segments $A, B,$ and $C$; see Fig.~\ref{fig:general-Cauchy}. Eq.~\eqref{eq:1111} is a generalization of the key equation in the main text; see Eq. (1) for comparison. Via a conformal mapping, we can also obtain the vacuum modular commutator for contiguous segments on an arbitrary Cauchy slice of the Lorentzian cylinder. If the Lorentzian cylinder has coordinates $X \in [-L/2, L/2), T \in (-\infty, \infty),$ with $U = T - X, V = T + X$, then our formula for the modular commutator reads:
\begin{equation}
    \label{eq:2222}
    \boxed{J(A,B,C)_{| \Omega_{\Sigma}\rangle} = \frac{\pi c_L}{6} (2\eta_V -1)-\frac{\pi c_R}{6} (2\eta_U -1),}
\end{equation}
where
\begin{equation}
    \eta_{V} = \frac{\sin(\pi V_{12}/L)\sin(\pi V_{34}/L)}{\sin (\pi V_{13}/L)\sin(\pi V_{24}/L)}, ~~ \eta_{U} = \frac{\sin(\pi U_{12}/L)\sin(\pi U_{34}/L)}{\sin (\pi U_{13}/L)\sin(\pi U_{24}/L)}.
\end{equation}
The two boxed equations are generalizations of the main result Eq.~\eqref{eq:main_result} and the first part of Eq.~\eqref{eq:eta-eff}. The second part of Eq.~\eqref{eq:eta-eff} is not particularly convenient to derive using the geometric formalism of this Section, since thermal states break local Lorentz invariance, which is one of the key tools of this Section. Thermal states are considered in the next Section using a different formalism.

Two ideas are used to derive the boxed equations. First, the modular flows for the subsystems we consider are local flows generated by vector fields within their respective domains of dependence. Thus, the modular flow can be treated as a transformation acting on Cauchy surfaces. To that end, we review known facts about modular flow in CFT in Section~\ref{appendix:mod_flow_cft}. Second, the modular commutator can be computed by calculating the response of the entanglement entropy against a modular flow. This calculation is explained for the Minkowski vacuum in Section~\ref{appendix:general_cauchy_surface}, and for the timelike cylinder vacuum in Section~\ref{appendix:finite-cylinder-flow}. For readers interested in the formal structure of the calculation, a brief review of some essential constructions in quantum field theory is provided in Section~\ref{appendix:qft-prelims}.

\subsection{States and surfaces in quantum field theory}
\label{appendix:qft-prelims}

In rigorous approaches to quantum field theory, a \textit{state} is defined abstractly as a collection of expectation values. More formally, a state $\omega$ is a linear map from operators to complex numbers that satisfies appropriate positivity and normalization conditions. (See e.g. \cite{witten2021does} for a review.) These operators need not be localized to a single time-slice --- for any local operator $\mathcal{O}$ at any spacetime position $(t, \vec{x}),$ the state $\omega$ has a corresponding expectation value $\omega(\mathcal{O}(t, \vec{x})).$ From the Hilbert space perspective, this is like the Heisenberg picture of quantum theory, where a state vector $| \Psi \rangle$ is defined globally for all times, and associates to a local operator $\mathcal{O}(t, \vec{x})$ the expectation value
\begin{equation}
    \omega_{\Psi}(\mathcal{O}(t, \vec{x})) \equiv \langle \Psi | \mathcal{O}(t, \vec{x}) | \Psi \rangle.
\end{equation}

This perspective is useful for field theory calculations, as it lets us talk about quantum states and their properties without ever having to pick a preferred time slice. However, it is sometimes useful to think about the restriction of a global state to a particular time slice --- or, more generally, to a particular Cauchy surface. So long as the field theory has a well posed initial value formulation, any operator in spacetime can be expressed in terms of the fundamental fields and their conjugate momenta on a single Cauchy surface. So restriction of a state to a single Cauchy surface contains all of the information of the global state. A Hilbert space can be constructed in terms of the fundamental fields and conjugate momenta on any Cauchy surface, allowing us to represent the global ``Heisenberg'' state $|\Psi\rangle$ as a slice-dependent ``Schrodinger'' state $|\Psi_{\Sigma} \rangle,$ for $\Sigma$ an arbitrary slice. This perspective is especially useful when approximating quantum field theories with lattice systems, where we \textit{do} have preferred notions of ``space'' and ``time,'' and spacetime emerges only in the continuum limit. While we are interested only in locally flat spacetimes, the perspective given here applies to general globally hyperbolic spacetimes.

For any set $S$ of spacetime points, its \textit{domain of dependence} $D(S)$ is the set of all spacetime points $p$ for which every inextendible causal path through $p$ intersects $S$. (In this language, a Cauchy surface is an achronal hypersurface whose domain of dependence is all of spacetime.) A generic set will be called a domain of dependence if it can be written as $D(\sigma)$ for some achronal hypersurface $\sigma$. There are, however, infinitely many hypersurface $\sigma$ which has the same domain of dependence $D(\sigma)$, see Figure~\ref{fig:general-Cauchy}(b) for an example. The restriction of a global quantum field theory state to any domain of dependence  defines a reduced density matrix for the observables restricted to that domain of dependence. For a well posed quantum field theory, we can equivalently think of this reduced density matrix as a state on the hypersurface $\sigma$; however, because the choice of $\sigma$ is non-unique, it is sometimes preferable in field theory calculations to think directly in terms of domains of dependence, rather than their spacelike slices. In the Schrodinger picture, one extends the hypersurface $\sigma$ to a full Cauchy surface $\Sigma,$ restricts the global state $|\Psi\rangle$ to its Schrodinger representative $|\Psi_{\Sigma}\rangle,$ and constructs $\rho_{\sigma}$ as the reduced density matrix of that vector restricted to the subsystem $\sigma.$

\subsection{Modular flow in CFT}
\label{appendix:mod_flow_cft}
Modular flow is the flow generated by the modular Hamiltonian. Let $A$ be a subregion of some Cauchy surface $\Sigma$ in 1+1D Minkowski spacetime. Once a state $|\Psi_{\Sigma}\rangle$ is specified on $\Sigma$, the modular Hamiltonian~\cite{Casini2011,cardy_entanglement_2016} associated with $A$ is $K_A= - \ln \rho_A$, where $\rho_A$ is the reduced density matrix of $|\Psi_{\Sigma}\rangle$ on $A$.\footnote{Strictly speaking, this operator is not defined in continuum quantum field theory; to be rigorous, one must work with the full modular theory reviewed e.g. in \cite{witten2018aps}. However, since entanglement entropies in CFT are only defined in regulated theories, all of our calculations are performed in regulated theories where the modular Hamiltonians are well defined. The continuum limit is only taken at the end of the calculation.} We shall focus on the vacuum state $|\Omega_{\Sigma}\rangle$ of the theory throughout this Section.\footnote{In the path integral approach to quantum field theory, one can think of $|\Omega_{\Sigma}\rangle$ as the state prepared by the path integral over the region in the past of $\Sigma.$}

The modular Hamiltonian generates the following one-parameter evolution of operators within the domain of dependence of $A$:
\begin{equation}
	\mathcal{O}_{A} \to U(-s) \mathcal{O}_{A} U(s), \quad \textrm{where} \quad U(s)= \rho_A^{is} = e^{- iK_A s}.
\end{equation}
For $s\in \mathbb{R}$, $U(s)$ is unitary. This transformation maps the set of operators in the domain of dependence to itself, thus preserving the algebra of operators on on $D(A)$. This map can generally be nonlocal, in the sense that local operators may be mapped to nonlocal operators.

However, for the cases considered below, the modular Hamiltonian is a local integral of the stress-energy tensor, so local operators do get mapped to local operators. Under the modular flow, we obtain $ U(s) \mathcal{O}_x U(-s)=\mathcal{O}'_{x(s)}$, where $\mathcal{O}'_{x(s)}$ is an operator at $x(s)$. In fact, if $\mathcal{O}$ is a primary, then $\mathcal{O}'$ is proportional to $\mathcal{O}$ up to a conformal factor \cite{Casini2011}. The spacetime point $x(s)$ can be obtained by integrating a vector field $2\pi V_K$ associated with the modular flow with the initial condition $x(0)=x$.

We discuss some examples below, accompanied by Fig.~\ref{fig:modular-flow-appendix} which depicts the vector fields $V_K$ of the first few examples.
\begin{exmp}[Quantum field theory and semi-infinite line]\label{ex:boost}

Let us begin with the case of a semi-infinite line $A=[0,\infty)$. The modular Hamiltonian of the vacuum restricted to this region is $K_A= 2\pi B$, 
where $B$ is the generator of the boost in the right Rindler wedge\footnote{The ``Rindler wedge'' is the domain of dependence $D(A)$ for $A$ a semi-infinite line.}. The corresponding vector field is
\begin{equation}
	V_K = x\partial_t + t\partial_x.
\end{equation}
In the light-cone coordinates $u=t-x$ and $v=t+x$, this becomes
\begin{equation}
	V_K = -u\partial_u + v \partial_v.
\end{equation}
This fact follows from Lorentz invariance alone, and as such, holds for any Lorentz invariant quantum field theory. This was derived axiomatically by Bisognano and Wichmann in \cite{BW1975, BW1976}, and using the path integral by Unruh and Weiss in \cite{unruh1984acceleration}. The modular flow is illustrated in Fig.~\ref{fig:modular-flow-appendix}(a).
\end{exmp}

\begin{figure}
	\centering
	\subfloat[Semi-infinite line]{\includegraphics[width=0.28\linewidth]{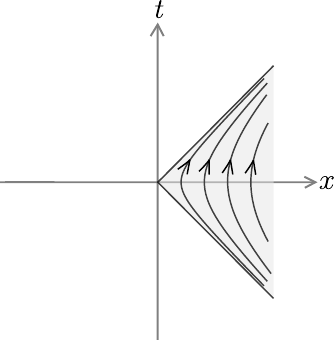}}
	\hspace{0.7cm}
	\subfloat[Finite spatial interval]{\includegraphics[width=0.28\linewidth]{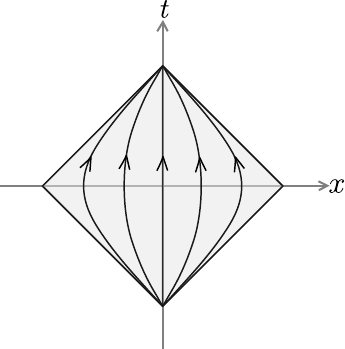}}
	\hspace{0.7cm}
	\subfloat[Finite spacetime interval]{\includegraphics[width=0.28\linewidth]{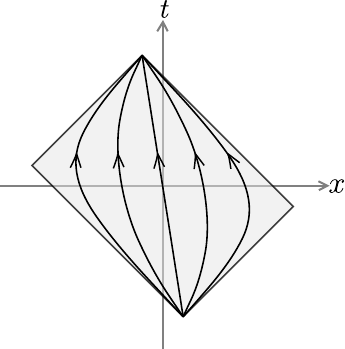}}
	\caption{Causal diamonds that are the domain of dependence of (a) a semi-infinite line, (b) a finite-line on the $x$-axis, and (c) a general spacelike interval. The modular flow for these cases can be represented as vector fields within the respective causal diamonds. }\label{fig:modular-flow-appendix}
\end{figure}

\begin{exmp}[1+1D CFT, a single finite interval in Minkowski spacetime]\label{ex:finite_interval}

Consider a 1+1D CFT in Minkowski spacetime and let $A$ be a finite spacetime interval, with endpoints $(u_1,v_1)$ and $(u_3,v_3)$. These two endpoints can be located on %a real line 
a single time-slice (Fig.~\ref{fig:modular-flow-appendix}(b)) or on a boosted interval (Fig.~\ref{fig:modular-flow-appendix}(c)). The conformal symmetry guarantees that the modular Hamiltonian is local and can be associated with a vector field $V_K$~\cite{Casini2011}. In light-cone coordinates, we obtain:
\begin{equation}
	\label{eq:VK_lorentz}
	V_K = -\frac{(u-u_1)(u_3-u)}{u_{13}}\partial_u+\frac{(v-v_1)(v_3-v)}{v_{13}}\partial_v,
\end{equation}  
where $u_{ij}\equiv u_j-u_i$.
From Eq.~\eqref{eq:VK_lorentz}, the expression  of $V_K$ in the $(x,t)$ coordinate can be worked out straightforwardly. For the particular case in which the two endpoints of the interval lie on the $x$ axis, with coordinates $(x_1,0)$ and $(x_3,0)$, we have
\begin{equation}
	V_K = \frac{(x-x_1)(x_3-x)-t^2}{x_{13}}\partial_t + \frac{t(x_1+x_3-2x)}{x_{13}}\partial_x.
\end{equation}
One way to derive Eq.~(\ref{eq:VK_lorentz}) is to use a conformal map from the Rindler wedge to the finite causal diamond shown in Fig.~\ref{fig:modular-flow-appendix}(c):
\begin{equation}\label{eq:global-conformal}
	f(u)= \frac{u_1-u_3 u}{1-u}, ~~ f(v) = \frac{v_1-v_3 v}{1-v}.
\end{equation}
Note that Eq.~\eqref{eq:global-conformal} is the analog of fractional linear transformations in the Lorentzian signature. 
\end{exmp}

Let us make a remark on a difference between these two examples, which will be useful in Section~\ref{appendix:general_cauchy_surface}. The vector field associated with the boost (Example~\ref{ex:boost}) preserves the spacetime volume element inside the Rindler wedge. That is to say, if we take an infinitesimal square-shaped diamond ($du dv$) inside the Rindler wedge, its volume is preserved, even thought the square becomes a rectangle. In contrast, the vector field in Example~\ref{ex:finite_interval} changes the spacetime volume element. This manifests in the fact that the upper half of the causal diamond is squashed under the flow.
\begin{exmp}[1+1D CFT, a single finite interval on Lorentzian cylinder] Consider a cylinder in Lorentzian signature, with spatial coordinate $X\in [-L/2,L/2)$ and time coordinate $T\in (-\infty,\infty)$. We also define light-cone coordinates $U=T-X$ and $V=T+X$. Minkowski spacetime can be conformally embedded into the cylinder via the map
\begin{equation}
\label{eq:Penrose}
    U = \frac{L}{\pi} \arctan u,~~ V = \frac{L}{\pi} \arctan v,
\end{equation}
where $u$ and $v$ are the Minkowski light-cone coordinates. This transformation maps Minkowski spacetime to the cylinder region with $U\in [-L/2,L/2)$ and $V\in [-L/2,L/2)$, see e.g. Fig.~\ref{fig:Wald-reproduced} for an illustration, and see Chapter 11 of \cite{Wald1984} for further details.\footnote{Since this map is conformal, it preserves causality; since its image has compact closure, some of its mathematical properties are simpler than those of Minkowski spacetime. This is a \textit{conformal compactification} of Minkowski spacetime, and was introduced by Penrose in \cite{penrose1965zero} to illustrate global causal properties of spacetime.} The modular flow Eq.~\eqref{eq:VK_lorentz} maps, under this transformation, to
\begin{equation} \label{eq:finite-cylinder-modular-flow}
    V_K = -\frac{L}{\pi}\frac{\sin(\pi(U-U_1)/L)\sin(\pi(U_3-U)/L)}{\sin(\pi U_{13}/L)}\partial_U+\frac{L}{\pi}\frac{\sin(\pi(V-V_1)/L)\sin(\pi(V_3-V)/L)}{\sin(\pi V_{13}/L)}\partial_V
\end{equation}
This gives the modular flow within the causal diamond on the Lorentzian cylinder bounded by $U_1 \geq U \geq U_3$ and $V_1\leq V \leq V_3$.
\end{exmp}

\begin{figure}[h]
	\centering
	\includegraphics[width=0.22\linewidth]{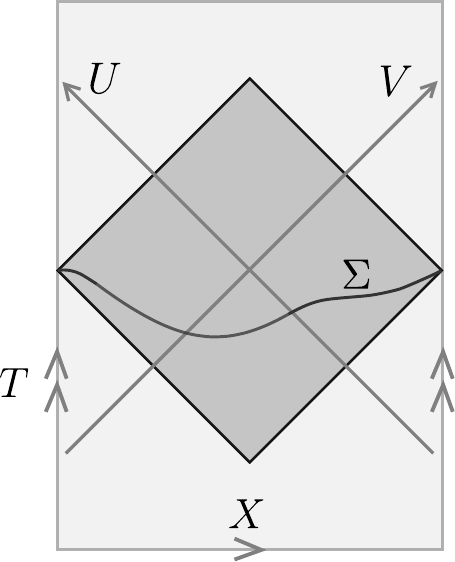}
	\caption{A Lorentzian cylinder in 1+1D,  where the left and right edges are identified. Two coordinate systems $(X,T)$ and $(U,V)$. $\Sigma$ represents a generic Cauchy surface. The dark gray diamond (which is not the domain of dependence of $\Sigma$) is obtained by a conformal embedding of the Minkowski spacetime to the cylinder.}\label{fig:Wald-reproduced}
\end{figure}

\subsection{Minkowski vacuum}
\label{appendix:general_cauchy_surface}

Let $\Sigma$ be a Cauchy surface in 1+1D Minkowski spacetime, and let $|\Omega_{\Sigma}\rangle$ be the Schrodinger-picture representation of the vacuum state on this surface. We can generalize the result in the main text to contiguous intervals $A,B,$ and $C$ on $\Sigma$. Let $(u_i,v_i)$, $i=1,2,3,4$ be the endpoints of the three intervals (see Fig.~\ref{fig:general-Cauchy}). (Here we assume that the space direction is infinite.) The key formula we derive in this Section is 
\begin{equation}\label{eq:J_general}
\boxed{	J(A,B,C)_{| \Omega_{\Sigma}\rangle } = \frac{\pi c_L}{6} (2\eta_v -1)-\frac{\pi c_R}{6} (2\eta_u -1),}
\end{equation}
where
\begin{equation}
	\eta_u = \frac{u_{12}u_{34}}{u_{13}u_{24}}, ~~\eta_v = \frac{v_{12}v_{34}}{v_{13}v_{24}}.
\end{equation}

\begin{figure}[h]
	\centering
	\includegraphics[width=0.68\linewidth]{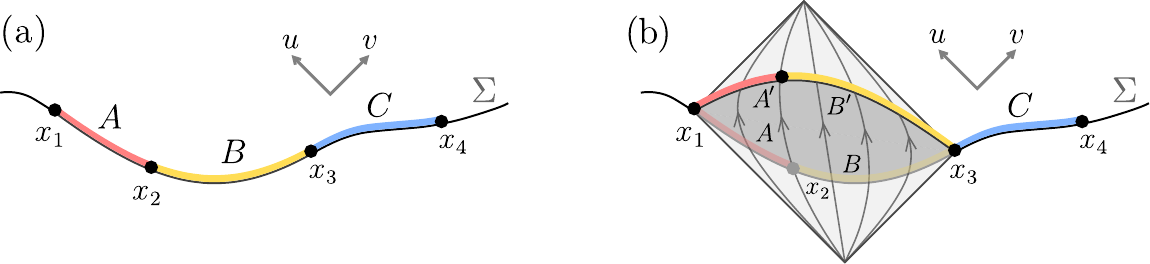}
	\caption{(a) Three contiguous intervals $A$, $B$ and $C$ on a general Cauchy surface $\Sigma$. (b) The modular flow within the domain of dependence (causal diamond) of $AB$. The coordinates of the endpoints are $x_i= (u_i, v_i)$.}\label{fig:general-Cauchy}
\end{figure}

\begin{remark}
	Before presenting the derivation, let us remark on the physical implications of Eq.~(\ref{eq:J_general}).
	\begin{itemize}
		\item The formula implies that one can determine both $c_L$ and $c_R$ by computing $J(A,B,C)$ for the state $|\Omega_{\Sigma}\rangle$ on a general Cauchy surface $\Sigma$. Note that the pair $\{c_L, c_R\}$ contains more information than the chiral central charge of the CFT, $c_-= c_L - c_R$. Since only the latter can be obtained from a computation of  $J(A,B,C)_{|\Omega_{\Sigma}\rangle }$ when $\Sigma$ is a constant-$t$ slice, a richer set of physical information may be obtained by considering more general Cauchy surfaces.
		\item		 Since the Cauchy surface is spacelike,  two points close in the $v$ coordinate must be also close in the $u$ coordinate. Therefore, if region $B$ becomes sufficiently small, i.e., $(u_2,v_2)\rightarrow (u_3,v_3)$, or equivalently $A$ and $C$ become sufficiently long, i.e., $(u_1,v_1)\rightarrow (-\infty,-\infty),(u_4,v_4)\rightarrow (+\infty,+\infty)$, we get $\eta_u = \eta_v =1$ and $J(A,B,C)= \pi c_{-}/6$. In the opposite limit where $A$ or $C$ becomes sufficiently small (or $B$ becomes sufficiently large), we have  $\eta_u = \eta_v =0$ and $J(A,B,C)= -\pi c_{-}/6$. These are independent of the shape of the Cauchy surface and so is the difference $J(\eta=1)-J(\eta=0)=\pi c_{-}/3$.
	\end{itemize}
\end{remark}

In the rest of this Section, we provide a geometric derivation of Eq.~\eqref{eq:J_general}, making use of the facts about modular flow discussed in Section~\ref{appendix:mod_flow_cft}. Specifically, we make use of the vector field associated with the modular flow; see Eq.~\eqref{eq:VK_lorentz}. Like in the main text, the modular commutator is given by the response of the entanglement entropy $S_{BC}$ under the modular flow \cite{Kim2021a}, where
\begin{equation} \label{eq:appendix-modular-flow-formula}
	J(A,B,C)_{\rho} = -\left.\frac{dS(\rho_{BC}(s))}{ds}\right|_{s=0},
\end{equation}
and
\begin{equation}
	\rho_{BC}(s) = \Tr_A (e^{-i  K_{AB}s} \rho_{ABC} e^{i  K_{AB}s}).
\end{equation}
The entanglement entropy of a spacelike interval in the vacuum of a chiral CFT~\cite{Iqbal2016} is 
\begin{equation} \label{eq:appendix-iqbal-wall}
	S_{BC} = \frac{c_L}{12}\ln\frac{(v_4-v_2)^2}{\epsilon_{v2}\epsilon_{v4}} + \frac{c_R}{12}\ln\frac{(u_4-u_2)^2}{\epsilon_{u2}\epsilon_{u4}}.
	\end{equation} 
Here $\epsilon_{u(v)2(4)}$ denotes the UV cutoffs in the $u$ and $v$ directions at the respective endpoints. 
The UV cutoffs satisfy $\epsilon_{u2}\epsilon_{v2}=\epsilon^2_2, \epsilon_{u4}\epsilon_{v4}=\epsilon^2_4$, where $\epsilon_{2(4)}$ is the proper length of the UV cutoff at the two points.

For concreteness, we will regulate our quantum field theory by excising small domains of dependence from our spacetime at the left and right endpoints of the segment $BC$. At the left endpoint, we specify a domain of dependence by requiring that (i) it is the domain of dependence of some segment of $BC$ containing the left endpoint, and (ii) its spacetime volume is $\epsilon_2^2$. At the right endpoint, we specify an analogous domain of dependence with spacetime volume $\epsilon_4^2.$ In null coordinates, the excised regions are
\begin{eqnarray}
    Q_{L} & = & [u_2, u_2 - \epsilon_{u 2}] \times [v_2, v_2 + \epsilon_{v 2}], \label{eq:cutoff-1} \\
    Q_{R} & = & [u_{4}, u_{4} + \epsilon_{u 4}] \times [v_4, v_4 - \epsilon_{v 4}]. \label{eq:cutoff-2}
\end{eqnarray}
Conformal transformations preserve the causal structure of a spacetime, so they map domains of dependence to domains of dependence. Since the modular flow of the $AB$ region corresponds to a local conformal transformation, it will deform the domains of dependence $Q_L$ and $Q_R$ to domains of dependence with different null extents $\epsilon'_{u(v)2(4)}.$ This change contributes nontrivially to Eq.~\eqref{eq:appendix-modular-flow-formula} via Eq.~\eqref{eq:appendix-iqbal-wall}.\footnote{Further discussion on anomalies of entanglement entropies in chiral CFTs due to cutoff rescaling can be found in Section~2.1 of \cite{Iqbal2016}.
}

So, the change of $S_{BC}$ under $AB$ modular flow comes from two contributions: 
(i) the change of the spacetime point $(u_2,v_2)$ under the flow, and (ii) the rescaling of UV cutoffs in light-cone coordinates at the same point $(u_2,v_2)$. The formula for this change is
\begin{equation}
	\label{eq:dS_general}
	\frac{dS_{BC}}{ds} = -\left(\frac{c_L}{6}\frac{dv_2/ds}{v_{24}}+\frac{c_L}{12}\frac{d\ln \epsilon_{v2}}{ds}\right)+\left(v\leftrightarrow u, c_L \leftrightarrow c_R \right).
\end{equation}
Since in the modular flow (Eq.~\eqref{eq:VK_lorentz}) $u$ and $v$ coordinates are transformed separately, it follows that $J(A,B,C)$ admits two contributions, one depending only on $v$ and the other depending only on $u$.

The quantities $dv_2/ds|_{s=0}$ and $du_2/ds|_{s=0}$ can be read off directly from equation Eq.~\eqref{eq:VK_lorentz} as $2\pi$ times the coefficients of $\partial_v$ and $\partial_u$ with the substitutions $u \mapsto u_2, v \mapsto v_2$. We obtain
\begin{eqnarray}
    \left. \frac{d v_2}{ds}\right|_{s=0}
        & = & 2 \pi \frac{v_{12} v_{23}}{v_{13}}, \label{eq:dv} \\
    \left. \frac{d u_2}{ds} \right|_{s=0}
        & = & - 2 \pi \frac{u_{12} u_{23}}{u_{13}}. \label{eq:du}
\end{eqnarray}
The change in the cutoffs is given by how the endpoints of the segments in Eq.~\eqref{eq:cutoff-1} transform under modular flow. The change in the endpoints $u_2$ and $v_2$ are already given by Eqs.~\eqref{eq:dv} and \eqref{eq:du}; the other endpoints change as
\begin{eqnarray}
    \left. \frac{d (v_2 + \epsilon_{v2})}{ds} \right|_{s=0}
        & = & \left. \frac{d v_2}{ds} \right|_{s=0}
            + 2 \pi \epsilon_{v2} \frac{v_{23} - v_{12}}{v_{13}} + O(\epsilon_{v2}^2), \\
    \left. \frac{d (u_2 - \epsilon_{u2})}{ds} \right|_{s=0}
        & = & \left. \frac{d u_2}{ds} \right|_{s=0}
            + 2 \pi \epsilon_{u2} \frac{u_{23} - u_{12}}{u_{13}} + O(\epsilon_{u2}^2).
\end{eqnarray}
The second term in the first expression gives the change in the cutoff $\epsilon_{v2}$ under $AB$ modular flow, at leading order in the cutoff. The second term in the second expression gives \textit{minus} the change in $\epsilon_{u2}$ under $AB$ modular flow. From these expressions, we can immediately compute
\begin{eqnarray}
    \left. \frac{d \ln{\epsilon_{v2}}}{ds} \right|_{s=0}
        & = & 2 \pi \frac{v_{23} - v_{12}}{v_{13}}, \label{eq:dv_cutoff} \\
    \left. \frac{d \ln{\epsilon_{u2}}}{ds} \right|_{s=0}
        & = & - 2 \pi \frac{u_{23} - u_{12}}{u_{13}}. \label{eq:du_cutoff}
\end{eqnarray}
Finally, substituting Eqs.~\eqref{eq:dv}, \eqref{eq:du}, \eqref{eq:dv_cutoff}, and \eqref{eq:du_cutoff} into Eq.~\eqref{eq:dS_general} we obtain the main result Eq.~\eqref{eq:J_general} in the continuum limit as all cutoffs vanish.

\begin{remark}
	If the modular flow does not induce a change of the spacetime volume of element $dudv$ on the local patch near $(u_2,v_2)$, then the change of cutoffs has the following interpretation:
   \begin{equation}\label{eq:diff-cutoffs-appendix}
	d \ln \epsilon_{v2} = -d\ln \epsilon_{u2} = d \chi, 
    \end{equation}
    where $d\chi$ is the infinitesimal boost angle that develops at $(u_2,v_2)$. This formula holds if the modular flow is a boost, as in Fig.~\ref{fig:modular-flow-appendix}(a). This is also the case for the modular flow in the finite causal diamond in Fig.~\ref{fig:modular-flow-appendix}(b), provided that $(u_2,v_2)$ lies on the $x$ axis. For a general point $(u_2,v_2)$ within the causal diamond depicted in Fig.~\ref{fig:modular-flow-appendix}(b) and (c), formula Eq.~\eqref{eq:diff-cutoffs-appendix} no longer holds. However, the change of cutoffs can still be computed with Eq.~\eqref{eq:dv_cutoff} and \eqref{eq:du_cutoff}.
\end{remark}

\subsection{Cylinder vacuum}
\label{appendix:finite-cylinder-flow}

The entanglement entropy of the vacuum on a timelike cylinder can be obtained from Eq.~\eqref{eq:appendix-iqbal-wall} by applying the conformal transformation Eq.~\eqref{eq:Penrose}. The formula for a spacelike interval is
\begin{equation} \label{eq:appendix-iqbal-wall2}
	S_{BC} = \frac{c_L}{12}\ln\frac{L^2 \sin^2(\pi V_{24}/L)}{\pi^2\epsilon_{V2}\epsilon_{V4}} + \frac{c_R}{12}\ln\frac{L^2\sin^2(\pi U_{24}/L)}{\pi^2 \epsilon_{U2}\epsilon_{U4}}.
\end{equation}
Consider three contiguous intervals $A$, $B$ and $C$ on a Cauchy surface $\Sigma$ of the cylinder; see Fig.~\ref{fig:Wald-reproduced} for an illustration of the Cauchy surface.
Under the geometric flow of the $AB$ region given in Eq.~\eqref{eq:finite-cylinder-modular-flow}, the formula for the change in entanglement entropy is
\begin{equation}
    \frac{dS_{BC}}{ds} = -\left(\frac{c_L}{6}\frac{(\pi/L)\times dV_2/ds}{\tan(\pi V_{24}/L)}+\frac{c_L}{12}\frac{d\ln \epsilon_{V2}}{ds}\right)+\left(V\leftrightarrow U, c_L \leftrightarrow c_R \right).
\end{equation}
Following the same logic as in the previous subsection, but with the modular flow given by $2\pi$ times the vector field in Eq.~\eqref{eq:finite-cylinder-modular-flow}, we obtain
\begin{eqnarray}
    \left. \frac{d V_2}{ds}\right|_{s=0}
        & = & 2L\frac{\sin(\pi V_{12}/L)\sin(\pi V_{23}/L)}{\sin(\pi V_{13}/L)}, \label{eq:dv_cyl} \\
    \left. \frac{d U_2}{ds} \right|_{s=0}
        & = & -2L\frac{\sin(\pi U_{12}/L)\sin(\pi U_{23}/L)}{\sin(\pi U_{13}/L)}. \label{eq:du_cyl}
\end{eqnarray}
and 
\begin{eqnarray}
    \left. \frac{d \ln{\epsilon_{V2}}}{ds} \right|_{s=0}
        & = & 2 \pi \frac{\sin(\pi (V_{23}-V_{12})/L)}{\sin(\pi V_{13}/L)}, \label{eq:dv_cutoff_cyl} \\
    \left. \frac{d \ln{\epsilon_{U2}}}{ds} \right|_{s=0}
        & = & - 2 \pi \frac{\sin(\pi (U_{23}-U_{12})/L)}{\sin(\pi U_{13}/L)}. \label{eq:du_cutoff_cyl}
\end{eqnarray}
Now we can substitute the derivatives into $J(A,B,C) = -(dS_{BC}/ds)|_{s=0}$. After a bit trigonometry we obtain
\begin{equation}
    J(A,B,C)_{|\Omega_{\Sigma}\rangle} = \frac{\pi c_L}{6} (2\eta_V -1)-\frac{\pi c_R}{6} (2\eta_U -1),
\end{equation}
where 
\begin{equation}
    \eta_{V} = \frac{\sin(\pi V_{12}/L)\sin(\pi V_{34}/L)}{\sin (\pi V_{13}/L)\sin(\pi V_{24}/L)}, ~~ \eta_{U} = \frac{\sin(\pi U_{12}/L)\sin(\pi U_{34}/L)}{\sin (\pi U_{13}/L)\sin(\pi U_{24}/L)}.
\end{equation}

\section{Operator-based derivation}
\label{appendix:operators}

We present an operator-based derivation of the general formulas Eq.~\eqref{eq:1111} and Eq.~\eqref{eq:2222}, as well as for the second part of Eq.~\eqref{eq:eta-eff}. The key point is that because the modular flow of regions $AB$ and $BC$ is local in all the cases we consider, the corresponding modular Hamiltonians $K_{AB}$ and $K_{BC}$ can be expressed as local integrals of the stress-energy tensor. The modular commutator can then be obtained via commutation relations for the left- and right-moving components of the stress-energy tensor. In Section~\ref{appendix:mod_ham_SE} we review how, in 1+1D CFT, the modular Hamiltonian of the vacuum state restricted to a spacelike segment can be written in terms of the stress-energy tensor. In Section~\ref{appendix:SE_commutators}, we review how self-commutators of the Minkowski stress-energy tensor can be obtained via the operator product expansion of the Euclidean theory. In Section~\ref{appendix:SE_modular_derivation}, we use these two ingredients to compute Eq.~\eqref{eq:1111} for the modular commutator. In Section~\ref{appendix:SE_finite_cylinder}, we sketch how the same kind of calculation can be used to reproduce Eq.~\eqref{eq:2222} for the modular commutator on a timelike cylinder. In Section~\ref{appendix:SE-thermal-state}, we sketch how the same kind of calculation can be used to reproduce the second part of Eq.~\eqref{eq:eta-eff} for a thermal state at $L=\infty.$

\subsection{The modular Hamiltonian and the stress-energy tensor}
\label{appendix:mod_ham_SE}

We define the stress-energy tensor using the conventional normalization for a 1+1D CFT. The classical stress-energy tensor for a field theory with action $S$ is defined by 
\begin{equation} \label{eq:SE-tensor-normalization}
    T_{\mu \nu}^{\text{classical}} = - \frac{4 \pi}{\sqrt{-g}} \frac{\delta S}{\delta g^{\mu \nu}},
\end{equation}
and the quantum stress-energy tensor $T_{\mu \nu}$ is defined via an appropriate renormalization of $T_{\mu \nu}^{\text{classical}}.$ For any smooth vector field $\xi^{\mu},$ the quantum stress-energy tensor obeys a Ward identity that controls how correlation functions of the theory respond to the local diffeomorphism generated by $\xi^{\mu}.$ With the normalization given in Eq.~\eqref{eq:SE-tensor-normalization}, this Ward identity is\footnote{For a pedagogical introduction to the origins of this identity, see \cite{sorce_2022}.}
\begin{equation}
    \delta_{\xi} \langle \mathcal{O} \rangle
        = \frac{i}{2 \pi} \left\langle \mathcal{T} \int d^{2} x T_{\mu \nu}(x) \partial^{\mu} \xi^{\nu}(x) \mathcal{O} \right\rangle,
\end{equation}
where $\mathcal{O}$ is an arbitrary operator, $\langle \mathcal{T} \cdot \rangle$ denotes a time-ordered vacuum expectation value, and $\delta_{\xi} \mathcal{O}$ denotes the linear change of that operator when it is pushed forward by along the vector field $\xi^{\mu}.$ If we take $\xi^{\mu}$ to be of the form $f(x) Y^{\mu},$ where $f(x)$ is any function on the spacetime that is well behaved at infinity, then integrating by parts gives 
\begin{equation}
    \delta_{f Y} \langle \mathcal{O} \rangle
        = - \frac{i}{2 \pi} \left\langle \mathcal{T} \int d^{2} x f(x) Y^{\nu}(x) \partial^{\mu} T_{\mu \nu}(x) \mathcal{O} \right\rangle.
\end{equation}
If $Y^{\nu}$ generates a local conformal transformation, then $(\partial^{\mu} Y^{\nu}) T_{\mu \nu}$ vanishes by tracelessness of the stress-energy tensor, and we can commute $Y^{\nu}(x)$ through the derivative $\partial^{\mu}$ to obtain
\begin{equation}
    \delta_{f Y} \langle \mathcal{O} \rangle
        = - \frac{i}{2 \pi} \left\langle \mathcal{T} \int d^{2} x f(x) \partial^{\mu} ( Y^{\nu}(x) T_{\mu \nu}(x)) \mathcal{O} \right\rangle.
\end{equation}
If we choose $\mathcal{O}$ to be of the form $\mathcal{A} \mathcal{B}$, where $\mathcal{A}$ is localized to a Cauchy surface $\Sigma$ and $\mathcal{B}$ is supported in the past of $\Sigma,$ then by choosing $f(x)$ to be a function that equals $1$ in a small neighborhood of $\Sigma$ and $0$ everywhere else in spacetime, we can apply the divergence theorem to obtain
\begin{equation}
    \langle (\delta_{Y} \mathcal{A}) \mathcal{B} \rangle
        = - \frac{i}{2 \pi} \left\langle \mathcal{T} \left( \int_{\Sigma + \epsilon} Y^{\nu} T_{\mu \nu} dS^{\mu} \right) \mathcal{A} \mathcal{B} \right\rangle
        + \frac{i}{2 \pi} \left\langle \mathcal{T} \left( \int_{\Sigma - \epsilon} Y^{\nu} T_{\mu \nu} dS^{\mu} \right) \mathcal{A} \mathcal{B} \right\rangle
\end{equation}
where $\Sigma + \epsilon$ and $\Sigma - \epsilon$ denote Cauchy surfaces slightly in the future or past of $\Sigma.$ The time-ordering of the correlation function ensures that in the limit $\epsilon \rightarrow 0,$ this identity becomes
\begin{equation}
    \langle (\delta_{Y} \mathcal{A}) \mathcal{B} \rangle
        = - \frac{i}{2 \pi} \left\langle \left[ \int_{\Sigma} Y^{\nu} T_{\mu \nu} dS^{\mu},  \mathcal{A} \right] \mathcal{B} \right\rangle.
\end{equation}
There is no time-ordering in this expression, as we have chosen $\mathcal{B}$ to be supported in the past of $\Sigma.$ Since this expression holds for arbitrary $\mathcal{B},$ it implies the operator identity
\begin{equation}
\label{eq:charge_transform}
    \delta_{Y} \mathcal{A} = - \frac{i}{2 \pi} \left[ \int_{\Sigma} Y^{\nu} T_{\mu \nu} dS^{\mu},  \mathcal{A} \right].
\end{equation}
This tells us that for \textit{any} vector field $Y^{\mu}$ that generates a local conformal transformation, its generator can be expressed as the operator
\begin{equation} \label{eq:local-flow-generator}
    - \frac{1}{2\pi} \int_{\Sigma} Y^{\nu} T_{\mu \nu} dS^{\mu}
\end{equation}
on any Cauchy surface $\Sigma.$ Note that the stress-energy tensor equations $\partial^{\mu} T_{\mu \nu} = T_{\mu}{}^{\mu} = 0$ imply that the current $J^{\mu} = Y^{\nu} T_{\mu \nu}$ is conserved for $Y^{\mu}$ the generator of a local conformal transformation. The conserved charge Eq.~\eqref{eq:local-flow-generator} is therefore a topological operator independent of $\Sigma$, that is, the current integrated over $\Sigma$ is the same as the current integrated over $\Sigma'$ except in a correlation function containing operators between $\Sigma$ and $\Sigma'.$ From this perspective, Eq.~\eqref{eq:charge_transform} is just the usual notion that conserved charges generate symmetry transformations, adapted to the stress-energy tensor.

In 1+1D Minkowski spacetime, the surface integral Eq.~\eqref{eq:local-flow-generator} can be written in the sense of differential forms as
\begin{equation} \label{eq:local-flow-generator-forms}
    - \frac{1}{2 \pi} \int_{\Sigma} Y^{\nu} T_{\mu \nu} dS^{\mu}
        = \frac{1}{2 \pi} \int_{\Sigma} \left( Y^{v} T_{v v} dv - Y^{u} T_{u u} du \right),
\end{equation}
where we have used the tracelessness condition $T_{uv}=0.$ As explained in Section~\ref{appendix:mod_flow_cft}, in the vacuum state of a 1+1D CFT, the modular flow of a spacelike segment $R$ with endpoints $(u_1, v_1), (u_2, v_2)$ is generated by the vector field $2 \pi V_R$ with
\begin{equation} \label{eq:VR-appendix-2}
    V_R = - \frac{(u - u_1)(u_2 - u)}{u_{12}} \partial_u + \frac{(v - v_1) (v_2 - v)}{v_{12}} \partial_{v},
\end{equation}
where $v = t+x$ and $u = t-x$ are light-cone coordinates. Since $2 \pi V_R$ generates a local conformal transformation, we can plug it into the expression in Eq.~\eqref{eq:local-flow-generator-forms} to obtain the modular Hamiltonian as
\begin{equation} \label{eq:KR-indicator}
    K_R = \int_{\Sigma} \left( V_R^{v} T_{v v} dv - V_R^{u} T_{u u} du \right) \Pi_{(u_1, v_1), (u_2, v_2)}
\end{equation}
where the indicator function $\Pi_{(u_1, v_1), (u_2, v_2)}$  restricts the modular flow to act only on the portion of $\Sigma$ lying between the given endpoints.

Ordinarily, to evaluate Eq.~\eqref{eq:KR-indicator}, we would need to write $u$ as $u = f(v)$ on $\Sigma,$ substitute $u \mapsto f(v), du \mapsto f'(v) dv,$ then integrate over $v$ from $(-\infty, \infty).$ But by inspection of Eq.~\eqref{eq:VR-appendix-2}, we see that $V_R^u$ depends only on $u,$ and $V_{R}^v$ depends only on $v.$ The stress-energy tensor conservation law $\partial^{\mu} T_{\mu \nu} = 0$ in null coordinates further implies that $T_{uu}$ depends only on $u$, while $T_{vv}$ depends only on $v.$ Finally, on the surface $\Sigma,$ the indicator function can be written in terms of either $u$ or $v$ alone:
\begin{equation}
    \Pi_{(u_1, v_1), (u_2, v_2)}
        = \Theta(u_1 - u) \Theta(u - u_2) = \Theta(v - v_1) \Theta(v_2 - v). 
\end{equation}
So the two integrals in Eq.~\eqref{eq:KR-indicator} each depend on only one of the variables $u$ and $v,$ and there is no need to make a restriction $u = f(v)$ to define the integral. (This can be viewed as a manifestation of the fact that $K_R$ is a topological operator, and does not depend on the specific shape of $\Sigma$.) We therefore have the general expression 
\begin{equation} \label{eq:mod-ham-null}
    K_R = \int_{-\infty}^{\infty} dv\, V_R^{v} \Theta(v - v_1) \Theta(v_2 - v) T_{v v} + \int_{-\infty}^{\infty} du\, V_R^{u} \Theta(u_1 - u) \Theta(u - u_2) T_{u u}.
\end{equation}
The change in the sign of the $u$ term as compared to Eq.~\eqref{eq:KR-indicator} comes from the fact that the the coordinate $v$ is co-oriented with the spatial coordinate $x,$ while the coordinate $u$ is counter-oriented against the spatial coordinate $x.$ I.e., as $x$ increases, $v$ increases but $u$ decreases.

Since the modular commutator is defined as
\begin{equation}
    J(A, B, C)_{|\Omega_{\Sigma}\rangle}
        = i \langle \Omega_{\Sigma} | [K_{AB}, K_{BC}] | \Omega_{\Sigma} \rangle,
\end{equation}
we can compute it directly using Eq.~\eqref{eq:mod-ham-null} when $A, B, C$ are contiguous, since in this case $AB$ and $BC$ are both connected segments. The only remaining ingredients we need are the commutators of $T_{uu}$ and $T_{vv},$ which we compute in the next subsection.

\subsection{Stress-energy commutators}
\label{appendix:SE_commutators}

The Euclidean theory is obtained from the Lorentzian theory via the substitution $t \mapsto - i \tau,$ i.e. $u \mapsto - i \tau - x \equiv - z, v \mapsto - i \tau + x \equiv \bar{z}.$ Under this substitution, we have 
\begin{equation}
    T_{uu} = T_{zz}, \qquad T_{vv} = T_{\bar{z} \bar{z}}.
\end{equation}
So we can obtain the self-commutators of $T_{uu}$ and $T_{vv}$ via the self-commutators of $T_{zz}$ and $T_{\bar{z} \bar{z}}.$ Strictly speaking, these identifications can only be made on the slice $t = \tau = 0.$ However, since $T_{uu}$ depends only on $u$ and $T_{vv}$ depends only on $v,$ they can always be represented on the $t=0$ slice, i.e., we have
\begin{align}
    T_{uu}(t, x)
        & = T_{uu}(t=0, x - t) = T_{zz}(\tau=0, x-t) \\
    T_{vv}(t, x) & =  T_{vv}(t=0, x + t) = T_{\bar{z} \bar{z}}(\tau=0, x + t).
\end{align}

The self-commutators of $T_{zz}$ and $T_{\bar{z} \bar{z}}$ on the $\tau=0$ slice can be computed in a standard way from the operator product expansion (OPE) for the stress tensor in the Euclidean version of the theory. We will present a self-contained version of this technique, but we refer the reader to Section 2 of \cite{besken2021local} for further details.\footnote{The self-commutator of the stress-energy tensor can also be computed directly in the Lorentzian field theory using microcausality and dimensional analysis. This is the subject of an influential unpublished manuscript by Luscher and Mack, a summary of which can be found in Section III of \cite{mack1988introduction}. Note that some constants differ between our Eqs. \eqref{eq:uu-commutator}, \eqref{eq:vv-commutator} and Mack's Eq.~(14). This is because our stress-energy tensor $T$ is related to Mack's stress-energy tensor $\Theta$ by $T = - \pi \Theta$; his stress-energy tensor satisfies $[\int Y^{\nu} \Theta_{\mu \nu} dS^{\mu}, \mathcal{A}] = - 2 i \delta_Y \mathcal{A}$, as compared to our Eq.~\eqref{eq:charge_transform}.} The OPE for $T_{zz}$ and $T_{\bar{z} \bar{z}}$ is
\begin{align}
    T_{zz}(z_1) T_{zz}(z_2)
        & = \frac{c_R/2}{(z_1 - z_2)^4} + \frac{2 T_{zz}(z_2)}{(z_1 - z_2)^2} + \frac{\partial_{z} T_{zz}(z_2)}{z_1-z_2} + \dots, \\
    T_{\bar{z}\bar{z}}(\bar{z_1}) T_{\bar{z}\bar{z}}(\bar{z_2}) 
        & = \frac{c_L/2}{(\bar{z_1} - \bar{z_2})^4} + \frac{2 T_{\bar{z}\bar{z}}(\bar{z_2})}{(\bar{z_1} - \bar{z_2})^2} + \frac{\partial_{\bar{z}} T_{\bar{z}\bar{z}}(\bar{z_2})}{\bar{z_1}-\bar{z_2}} + \dots
\end{align}
Note that the chiral central charges $c_L$ and $c_R$ are switched as compared to the OPEs written in many texts, because we have defined $z$ by $z = x + i \tau$ rather than $z = \tau + i x.$

The OPE represents an equality of operators within arbitrary Euclidean-time-ordered correlators. So the expectation value of the commutator $[T_{zz}(0, x_1), T_{zz}(0, x_2)]$ can be obtained from the OPE as
\begin{equation}
    \langle [T_{zz}(0, x_1), T_{zz}(0, x_2)] \rangle
        = \lim_{\epsilon \rightarrow 0^+}
            \langle T_{zz}(\tau=\epsilon, x_1) T_{zz}(\tau=-\epsilon, x_2) - T_{zz}(\tau=-\epsilon, x_1) T_{zz}(\tau=\epsilon, x_2) \rangle,
\end{equation}
where Euclidean time-ordering is implicit in expression on the right side of this equation. Plugging in the OPE, we obtain
\begin{align}
    \langle [T_{zz}(0, x_1), T_{zz}(0, x_2)] \rangle
        = \lim_{\epsilon \rightarrow 0^+}
            \left( \frac{c_R / 2}{(x_1 - x_2 + 2 i \epsilon)^4} + \frac{2 T_{zz}(0, x_2)}{(x_1 - x_2 + 2 i \epsilon)^2} + \frac{\partial_{z} T_{zz}(0, x_2)}{x_1 - x_2 + 2 i \epsilon} - (\epsilon \leftrightarrow -\epsilon) \label{eq:OPE-commutator} \right).
\end{align}
Using the identity
\begin{equation}
	\lim_{\epsilon \rightarrow 0^+} \left( \frac{1}{x-y-2 i \epsilon}- \frac{1}{x - y + 2 i\epsilon} \right) = 2\pi i \, \delta(x - y)
\end{equation}
and its derivatives, we can simplify Eq.~\eqref{eq:OPE-commutator} to obtain
\begin{equation} \label{eq:zz-commutator}
    \langle [T_{zz}(0, x_1), T_{zz}(0, x_2)] \rangle
        = \frac{i \pi c_R}{6} \partial_{x_1}^3 \delta(x_1 - x_2) + 4 \pi i \langle T_{zz}(0, x_2) \rangle \partial_{x_1} \delta(x_1 - x_2) - 2 \pi i \langle \partial_{z} T_{zz}(0, x_2) \rangle \delta(x_1 - x_2).
\end{equation}
An analogous calculation gives the $T_{\bar{z}\bar{z}}$ commutator as
\begin{equation} \label{eq:zbarzbar-commutator}
    \langle [T_{\bar{z}\bar{z}}(0, x_1), T_{\bar{z}\bar{z}}(0, x_2)] \rangle
        = - \frac{i \pi c_L}{6} \partial_{x_1}^3 \delta(x_1 - x_2) - 4 \pi i \langle T_{\bar{z}\bar{z}}(0, x_2) \rangle \partial_{x_1} \delta(x_1 - x_2) + 2 \pi i \langle \partial_{\bar{z}} T_{\bar{z}\bar{z}}(0, x_2) \rangle \delta(x_1 - x_2).
\end{equation}
On the $t=0$ slice, we may make the identifications $x = -u = v,$ which gives the self-commutators of the Minkowski stress-energy tensor as
\begin{align}
    \langle [T_{uu}(u_1), T_{uu}(u_2)] \rangle
        & = - \frac{i \pi c_R}{6} \partial_{u_1}^3 \delta(u_1 - u_2) - 4 \pi i \langle T_{uu}(u_2) \rangle \partial_{u_1} \delta(u_1 - u_2) + 2 \pi i \langle \partial_{u} T_{uu}(u_2) \rangle \delta(u_1 - u_2), \label{eq:uu-commutator} \\
    \langle [T_{vv}(v_1), T_{vv}(v_2)] \rangle
        & = - \frac{i \pi c_L}{6} \partial_{v_1}^3 \delta(v_1 - v_2) - 4 \pi i \langle T_{vv}(v_2) \rangle \partial_{v_1} \delta(v_1 - v_2) + 2 \pi i \langle \partial_{v} T_{vv}(v_2) \rangle \delta(v_1 - v_2). \label{eq:vv-commutator}
\end{align}

Note that cross-commutators between $T_{uu}$ and $T_{vv}$ vanish, since cross-commutators between $T_{zz}$ and $T_{\bar{z} \bar{z}}$ vanish, as the OPE of these operators has no singular terms.

\subsection{The Minkowski vacuum}
\label{appendix:SE_modular_derivation}

We are now prepared to compute the modular commutator directly in CFT. Let $A,$ $B,$ and $C$ be contiguous segments on an arbitrary Cauchy slice with endpoints $(u_1, v_1), (u_2, v_2), (u_3, v_3), (u_4, v_4).$ Keep in mind that we have $u_1 > u_2 > u_3 > u_4$ and $v_1 < v_2 < v_3 < v_4.$ From Eq.~\eqref{eq:mod-ham-null}, we can write the modular commutator as
\begin{align}
    J(A, B, C)_{|\Omega_{\Sigma}\rangle}
        = & i \int_{-\infty}^{\infty} dv\, \int_{-\infty}^{\infty} dv'\, f_{13}(v) f_{24}(v') \langle [T_{vv}(v), T_{vv}(v')] \rangle \nonumber \\
        & + i \int_{-\infty}^{\infty} du\, \int_{-\infty}^{\infty} du'\, g_{13}(u) g_{24}(u') \langle [T_{u u}(u), T_{uu}(u')] \rangle.
\end{align}
where we have introduced the notation
\begin{align}
    f_{jk}(v)
        & = \frac{(v - v_j)(v_k - v)}{v_{k} - v_j} \Theta(v - v_j) \Theta(v_k - v), \\
    g_{jk}(u)
        & = - \frac{(u - u_j)(u_k - u)}{u_{k} - u_j} \Theta(u_j - u) \Theta(u - u_k).
\end{align}
Plugging in Eqs.~\eqref{eq:uu-commutator} and \eqref{eq:vv-commutator} for the stress-tensor commutators, and using the fact that in the Minkowski vacuum the one-point functions of the stress-energy tensor and its derivative both vanish, we obtain the identity
\begin{align} \label{eq:penultimate-OPE-calc}
    J(A, B, C)_{|\Omega_{\Sigma}\rangle}
        = & \frac{\pi c_L}{6} \int_{-\infty}^{\infty} dv\, \int_{-\infty}^{\infty} dv'\, f_{13}(v) f_{24}(v') \partial_v^3 \delta(v - v') \nonumber \\
        & + \frac{\pi c_R}{6} \int_{-\infty}^{\infty} du\, \int_{-\infty}^{\infty} du'\, g_{13}(u) g_{24}(u') \partial_u^3 \delta(u - u').
\end{align}
These integrals can be evaluated by integration by parts, using the identities
\begin{align}
    f'_{jk}(v)
        & = \frac{v_j + v_k - 2 v}{v_k - v_j} \Theta(v - v_j) \Theta(v_k - v), \\
    f''_{jk}(v)
        & = \delta(v - v_j) + \delta(v - v_k) - \frac{2}{v_k - v_j} \Theta(v - v_j) \Theta(v_k - v), \\
    g'_{jk}(u)
        & = - \frac{u_j + u_k - 2 u}{u_k - u_j} \Theta(u_j - u) \Theta(u - u_k), \\
    g''_{jk}(u)
        & = \delta(u - u_j) + \delta(u - u_k) + \frac{2}{u_k - u_j} \Theta(u_j - u) \Theta(u - u_k).
\end{align}
The result is
\begin{equation}
    J(A, B, C)_{|\Omega_{\Sigma}\rangle}
        = \frac{ \pi c_L}{6} (2 \eta_v - 1) - \frac{ \pi c_R}{6} (2 \eta_u - 1).
\end{equation}
with $\eta_u = (u_{12} u_{34}) / (u_{13} u_{24})$ and $\eta_{v} = (v_{12} v_{34}) / (v_{13} v_{24}),$ in perfect agreement with Eq.~\eqref{eq:1111}. The relative sign difference between the two terms comes from the fact that the $v$ integral gets restricted by step functions to the interval $[v_2, v_3],$ while the $u$ integral gets restricted to the interval $[u_3, u_2].$

\subsection{The finite cylinder vacuum}
\label{appendix:SE_finite_cylinder}

When Minkowski spacetime is replaced by a timelike cylinder of circumference $L,$ Eqs.~\eqref{eq:uu-commutator} and \eqref{eq:vv-commutator} still hold under the substitution $u \mapsto U, v \mapsto V,$ with $U$ and $V$ the light-cone coordinates on the cylinder. The expression for the modular Hamiltonian given in Eq.~\eqref{eq:mod-ham-null} also holds with $u \mapsto U$ and $v \mapsto V,$ but with the vector field $V_R$ replaced by the appropriate one for modular flow on the cylinder: 
\begin{equation} \label{eq:cylinder-modular-flow-vector}
    V_R = -\frac{L}{\pi}\frac{\sin(\pi(U-U_1)/L)\sin(\pi(U_2-U)/L)}{\sin(\pi U_{12}/L)}\partial_U+\frac{L}{\pi}\frac{\sin(\pi(V-V_1)/L)\sin(\pi(V_2-V)/L)}{\sin(\pi V_{12}/L)}\partial_V.
\end{equation}
This agrees with the Cardy-Tonni expression for the modular Hamiltonian \cite{cardy_entanglement_2016}. The only other difference in the calculation is that on the cylinder, the expectation values $\langle T_{UU}(U) \rangle$ and $\langle T_{VV}(V) \rangle$ do not vanish, but instead have the Casimir values
\begin{align}
    \langle T_{UU}(U) \rangle
        & = - \left(\frac{2 \pi}{L} \right)^2 \frac{c_R}{24}, \\
    \langle T_{VV}(V) \rangle
        & = - \left(\frac{2 \pi}{L} \right)^2 \frac{c_L}{24}.
\end{align}
This leads to extra terms in the integrands of Eq.~\eqref{eq:penultimate-OPE-calc} proportional to $\partial_v \delta(v - v')$ and $\partial_u \delta(u - u').$ We omit the details, but one can verify that repeating the calculation of the previous subsection with these substitutions reproduces Eq.~\eqref{eq:2222} for the modular commutator on the cylinder.

\subsection{Thermal states on a time slice}
\label{appendix:SE-thermal-state}

In Minkowski spacetime, in a thermal state of temperature $\beta,$ the modular Hamiltonian of the interval $[x_1, x_2]$ at $T=0$ is~\cite{cardy_entanglement_2016}
\begin{equation}
    K_{[x_1, x_2]}
        = \frac{\beta}{\pi} \int_{x_1}^{x_2} dx\, \frac{\sinh(\pi(x_2-x)/\beta) \sinh(\pi(x-x_1)/\beta)}{\sinh(\pi x_{12}/\beta)} (T_{zz}(x) + T_{\bar{z} \bar{z}}(x)).
\end{equation}
This matches the expression for the cylinder modular Hamiltonian generated by Eq.~\eqref{eq:cylinder-modular-flow-vector} on the $t=0$ slice up to the substitutions $U \leftrightarrow u, V \leftrightarrow v, L \leftrightarrow \beta,$ and $\sin \leftrightarrow \sinh.$ Furthermore, the one-point functions of the stress tensor in a Minkowski thermal state match those of a finite-cylinder vacuum state up to the substitution $L \leftrightarrow \beta.$ So we see that the modular commutator on a $t=0$ slice of a Minkowski thermal state must exactly match the modular commutator of a $t=0$ slice on the finite cylinder under these substitutions, which gives us the second part of Eq.~\eqref{eq:eta-eff}.

\section{Topological invariance: bulk and edge of 2+1D topological orders}

We present the details on an information-theoretic argument used in the main text.  
For concreteness, we assume that the system is on a 2D disk, with an edge.  On the ground state $| \psi_{\rm 2D}\rangle$ we assume the two axioms of entanglement bootstrap~\cite{Shi2020}; see Fig.~\ref{fig:EB-setup-axioms}.

\begin{figure}[h]
	\centering
	\begin{tikzpicture}
		\node[] (A) at (0,0) {\includegraphics[scale=1.2]{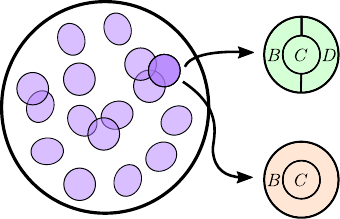}};
		\node[] (A) at (5.8,0.6) {{\bf A1:} $(S_{BC}+S_{CD}-S_B-S_D)_{| \psi_{\rm 2D}\rangle} =0.$};
		\node[] (A) at (5.2,-0.9) {{\bf A0:} $(S_{BC}+S_{C}-S_B)_{| \psi_{\rm 2D}\rangle} =0.$};
	\end{tikzpicture}
	\caption{The physical system on a disk and the two bulk-axioms ({\bf A0} and {\bf A1}) of entanglement bootstrap.}\label{fig:EB-setup-axioms}
\end{figure}

These axioms are applicable when the entanglement area law \cite{Kitaev2006,Levin2006} of 2+1D gapped system holds. 
For systems with a gapless edge --- which are of interest to us --- the area law is expected to hold when the regions are separated from the edge by a distance large compared to the bulk correlation length.\footnote{For realistic systems, especially those with ungappable edges, we expect the area law to have corrections which decay exponentially when the subsystems we pick are large compared to the correlation length. We drop these corrections in this Section, assuming these errors will not affect the analysis.} 

\begin{figure}[h]
	\centering
	\includegraphics[width=0.750\linewidth]{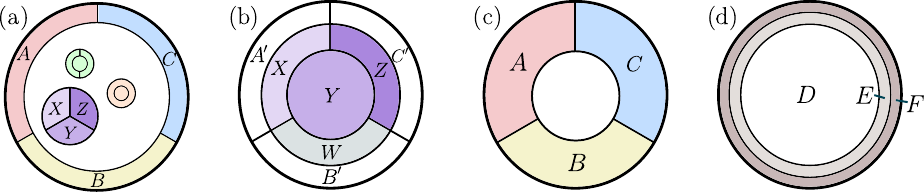}
	\caption{A 2+1D gapped chiral system on a disk and various choices of subsystems. (a) $ABC$ is an annulus covering the entire edge. The sizes for subsystems are large compared to the bulk correlation length. (b) A different partition of the disk, where $XYZW$ is a bulk disk. (c) Relabeling of regions in (b) as $A=A'X$, $B=B'W$, $C=C'Z$. (d) A finer partition of the annulus covering the edge; here $EF=A'B'C'$.}
	\label{fig:edge-appendix}
\end{figure}

\begin{Proposition} \label{prop:EB-1}
	The following statements about $J(A,B,C)$ are true. Statement 1, 2, and 3 follow from bulk-axiom {\bf A1}, whereas statement 4 needs both {\bf A0} and {\bf A1}.
	\begin{enumerate}
		\item $J(A,B,C)_{| \psi_{\rm 2D}\rangle}$ is invariant under any bulk deformation of the regions $A$, $B$, $C$. Here a bulk deformation is a ``smooth" deformation that preserves the topology of the union of any regions and does not add (or remove) any sites near the edge. 
		\item For the subsystems shown in Fig.~\ref{fig:edge-appendix}(a),
		\begin{equation}
			J(A,B,C)_{|\psi_{\rm 2D}\rangle}= - J(X,Y,Z)_{|\psi_{\rm 2D}\rangle}, \label{eq:same-appendix}
		\end{equation}
		
		\item Let $\vert \widetilde{\psi}_{\rm 2D} \rangle:= U_{\rm edge} | \psi_{\rm 2D}\rangle$, where $U_{\rm edge}$ is an arbitrary unitary operator supported on $A'B'C'$ of  Fig.~\ref{fig:edge-appendix}(b). Then
		\begin{equation}\label{eq:U-appendix}
			J(A,B,C)_{|\widetilde{\psi}_{\rm 2D}\rangle}=  J(A,B,C)_{|\psi_{\rm 2D}\rangle}.
		\end{equation}
		\item Let $\vert \widetilde{\psi}_{\rm 2D} \rangle:= \lambda \cdot O_{\rm edge} | \psi_{\rm 2D}\rangle$, where $O_{\rm edge}$ is an operator supported near the edge, on region $F$ (a finer subset of $A'B'C'$); see  Fig.~\ref{fig:edge-appendix}(d). ($O_{\rm edge}$ does not annihilate $| \psi_{\rm 2D}\rangle$, and $\lambda$ is a constant which normalizes the state.) Then $J(A,B,C)_{|\widetilde{\psi}_{\rm 2D}\rangle}=  J(A,B,C)_{|\psi_{\rm 2D}\rangle}$.
	\end{enumerate}
\end{Proposition}

\begin{remark}
	A few remarks are in order.
	\begin{enumerate}
	    \item The first statement in Proposition~\ref{prop:EB-1} does not imply any equality of the modular commutator of the following two configurations:
		\begin{equation}
			\includegraphics[scale=1.25]{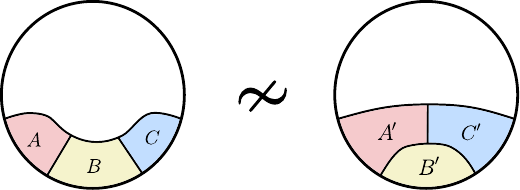}
		\end{equation}
		These two configures $ABC$ and $A'B'C'$ cannot be smoothly deformed into each other. Indeed, the modular commutators for these two choices of regions are not equal, in general.
		
		\item  The second statement essentially follows from the analysis in Section VI of~\cite{Kim2021a}.
		
	%	\item While axiom {\bf A1} is used in all the statements, axiom {\bf A0} is used only in the proof of the fourth statement.
	\end{enumerate}
\end{remark}

Here are the mathematical facts that will be relevant for our derivation. We define the conditional mutual information as $I(X:Z|Y)_{\rho}:= S(\rho_{XY}) + S(\rho_{YZ}) - S(\rho_Y) - S(\rho_{XYZ})$. Strong subadditivity (SSA) \cite{Lieb1973} refers to the statement that $I(X:Z|Y)_{\rho}\ge 0$ for any tripartite state $\rho_{ABC}$.
Petz showed that $I(X:Z|Y)_{\rho}=0$ if and only if $\ln \rho_{XYZ} = \ln \rho_{XY} + \ln \rho_{YZ} - \ln \rho_Y$~\cite{petz_monotonicity_2003}.  For later convenience, we introduce notation $K_X\equiv -\ln \rho_X$ for the modular Hamiltonian. Petz's result can then be written as
\begin{equation}
	I(X:Z|Y)_{\rho}=0 \quad \Leftrightarrow \quad K_{XYZ} = K_{XY} + K_{YZ} - K_Y.
\end{equation}
An immediate consequence is that 
\begin{equation} \label{eq:QMS-implication-J}
	I(X:Z|Y)_{\rho} = 0 \quad \Rightarrow \quad J(X,Y,Z)_{\rho}=0.
\end{equation}
For any bipartite pure state, $\ln \rho_X \vert \Psi_{XY}\rangle = \ln \rho_Y \vert \Psi_{XY}\rangle$, where $\rho_X$ and $\rho_Y$ are the reduced density matrices of $| \Psi_{XY}\rangle$. Another useful equality is $i\langle[K_{AB},K_A]\rangle=0$ for any mixed state $\rho_{AB}$, that is if one subsystem is part of the other subsystem then the modular commutator vanishes. This is because $\Tr(\rho_{AB}[K_{AB},K_A])=\Tr([\rho_{AB},K_{AB}]K_A)$ using the cyclic property of trace.
Moreover, the following two lemmas are useful.
\begin{lemma}\label{lemma:A1-extension}
    Given a state $\rho_{BCD}$ such that 
    \begin{equation}
    \label{eq:A1-in-lemma}
    S_{BC} + S_{CD} - S_B - S_D=0,
    \end{equation}
     any extension $\rho_{BB'CDD'}$ satisfies
    \begin{equation}
        S_{BB'C} + S_{CDD'} - S_{BB'} - S_{DD'}  =0.
    \end{equation}
    %If axiom \textbf{A1} applies to a configuration of regions $B : C : D$ in Fig.~\ref{fig:EB-setup-axioms}, then it also applies to any extension $B B' : C : D D'.$
\end{lemma}
\begin{proof}
    Let $\rho_{B B' C D D'}$ be the ``extended'' state, and let $|\psi\rangle_{B B' C D D' E}$ be a purification. %Axiom \textbf{A1} applied to $B : C : D$ is the statement
    %\begin{equation} \label{eq:A1-in-lemma}
    %    S_{BC} + S_{CD} - S_{B} - S_{D} = 0.
    %\end{equation}
    Strong subadditivity implies the inequalities
    \begin{align}
        S_{B B' C} - S_{B B'} - S_{B C} + S_{B} & \leq 0, \\
        S_{C D D'} - S_{D D'} - S_{CD} + S_{D} & \leq 0.
    \end{align}
    Adding these inequalities, and using Eq.~\eqref{eq:A1-in-lemma}, gives
    \begin{equation}
        S_{B B' C} + S_{C D D'} - S_{B B'} - S_{D D'} \leq 0.
    \end{equation}
    Applying purity of the global state to the second and third terms, we can rewrite this as
    \begin{equation}
        0 \geq S_{B B' C} + S_{C D D'} - S_{B B'} - S_{D D'}
            = S_{B B' C} + S_{B B' E} - S_{B B'} - S_{B B' C E}.
    \end{equation}
    But the rightmost expression in this equality is nonnegative by strong subadditivity applied to the tripartition $C : B B' : E.$ So we conclude
    \begin{equation}
        S_{B B' C} + S_{C D D'} - S_{B B'} - S_{D D'} = 0,
    \end{equation}
    as desired.
\end{proof}

\begin{lemma}\label{lemma:deformation}
	The following statements about the modular commutator are true. 
	\begin{enumerate}
		\item If $\rho_{aABC}$ satisfies $I(a:B|A)=0$, then $J(aA,B,C)=J(A,B,C)$. 
		\item If $\rho_{ABbC}$ satisfies $I(A:b|B)=I(C:b|B)=0$, then $J(A,bB,C)=J(A,B,C)$.
	\end{enumerate}
\end{lemma}
\begin{proof}
	
	First, we prove the first statement.
	$I(a: B \mid A)=0$ implies $K_{a A B}=K_{a A}+K_{A B}-K_{A}$
	Therefore,
	\begin{equation}
		\begin{aligned}
			J(aA,B,C) &=	i \langle \left[K_{a A B}, K_{B C}\right] \rangle\\
			&=i\left[K_{a A}+K_{A B}-K_{A}, K_{BC} \right]\rangle \\
			&=i\left\langle\left[K_{A B}, K_{B C} \right]\right\rangle \\
			%	\Rightarrow J(a A, B, C)=J(A, B, C) \\
			&= J(A,B,C).
		\end{aligned}
	\end{equation}
	
	Next, we prove the second statement. By $I(A:b \mid B)=I(C:b \mid B)=0$, we have
	%	Proof of $J(A, B b, C)=J(A, B, C)$, assuming
	\begin{equation}
		\label{eq:Petz_appendix}
		\begin{aligned}
			K_{A B b}&=K_{A B}+K_{B b}-K_{B}  \\
			K_{C B b}&=K_{C B}+K_{B b}-K_{B} .
		\end{aligned}
	\end{equation}
	Using these two equations to replace bigger chunks of modular Hamiltonians by small ones supported on the marginals, we get
	\begin{equation*}
		\begin{aligned}
			J(A, B b, C)=& i\left\langle\left[ K_{A B b}, K_{BC b}\right]\right\rangle\\		
			=& i\langle[K_{A B}+K_{B b}-K_{B},  K_{BC}+K_{B b}-K_{B}]\rangle \\
			=& i\langle[K_{A B}, K_{B C}]\rangle+i\langle[K_{A B}, K_{B b}]\rangle +i\left\langle\left[K_{B b}, K_{BC}\right]\right\rangle
			\\			
			& -i\langle [K_{A B}, K_{B}] \rangle -i\left\langle\left[K_{B b}, K_{B}\right]\right\rangle -i\left\langle\left[K_{B}, K_{BC}\right]\right\rangle-i\left\langle\left[K_{B}, K_{B b}\right]\right\rangle \\			
			=& i\left\langle\left[K_{A B}, K_{B C}\right]\right\rangle\\
			=& J(A,B,C)
		\end{aligned}
	\end{equation*}
	Two of the three terms in the third line, namely $i\left\langle\left[K_{A B}, K_{B b}\right]\right\rangle$ and $i\left\langle\left[K_{B b}, K_{BC}\right]\right\rangle$ vanish upon using Eq.~\eqref{eq:QMS-implication-J}.
	Terms in the fourth line, e.g., $i\left\langle\left[K_{B}, K_{BC}\right]\right\rangle$ and $i\left\langle\left[K_{B}, K_{B b}\right]\right\rangle$ vanish because $B$ is a subsystem of $BC$ and $Bb$. 
\end{proof}

\begin{proof}[Proof of Proposition~\ref{prop:EB-1}]

	First, we prove the invariance of the modular commutator under deformations of boundary regions $A,$ $B,$ and $C$. Because of the symmetry of the modular commutator $J(X,Y,Z)= -J(Z,Y,X)$, it suffices to consider the bulk deformations shown in Fig.~\ref{fig:bulk-deformations}. In total, there are 5 cases.
	\begin{itemize}
		\item For (a), we wish to show $J(aA,B,C)=J(A,B,C)$.
		Axiom \textbf{A1} applies to the ``thickening'' of $a$ depicted in Fig.~\ref{fig:bulk-deformations}(a), so Lemma~\ref{lemma:A1-extension} tells us that it also applies to the system $A:a:CD.$ The purity relations $S_{aCD}=S_{AB}$ and $S_{CD} = S_{a A B}$ then imply $I(a:B|A) = 0,$ at which point we can apply Lemma~\ref{lemma:deformation}. 
		\item For (b), we wish to show $J(aA,B\setminus a,C)=J(A,B,C)$.
		To derive this identity, we use the fact that $| \psi_{ABCD}\rangle $ is pure to convert the problem into showing $J(aA,D,C)=J(A,D,C)$, which we already solved in (a).
		\item For (c), we wish to show $J(A,Bb,C)=J(A,B,C)$.
		The bulk axiom {\bf A1} applied to a thickening of $b$ can be extended, via Lemma~\ref{lemma:A1-extension}, to the systems $B:b:CD$ and $B:b:AD$. This implies, via the purity relations $S_{bCD} = S_{AB},$ $S_{CD} = S_{A B b}$, $S_{b A D} = S_{B C},$ $S_{A D} = S_{b B C}$, the identity $I(b:A|B)=I(b:C|B)=0$. The modular commutator identity then follows immediately from Lemma~\ref{lemma:deformation}.
		\item For (d), we wish to show $J(aA,B,C)=J(A,B,C)$. 
		Using the purity of $| \psi_{ABCD}\rangle $, we can convert the problem into the identity $J(D,C,B)=J(D\setminus a,C, B)$. This follows from Lemma~\ref{lemma:deformation} and $I(a:C | D\setminus a)=0,$ which follows from the same kind of ``thickening'' trick as in the previous examples.
		\item For (e), we wish to show $J(A,Bb,C)=J(A,B,C)$. Using purity, we convert the problem into showing $J(D\setminus b , C, Bb)= J(D,C,B)$. This follows from $I(b:C | D\setminus b) = I(b:C|B)=0,$ which again follows from a ``thickening'' trick, and using Lemma~\ref{lemma:deformation} twice.
		 
	\end{itemize}

	\begin{figure}[h]
		\centering
	\includegraphics[scale=1.5]{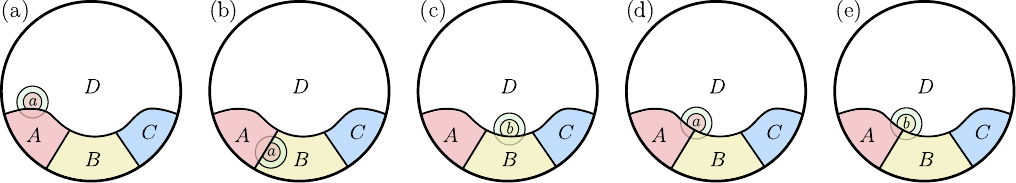}
		\caption{Bulk-deformations. Each case is verified using axiom {\bf A1} on a bulk disk.}\label{fig:bulk-deformations}
	\end{figure}

	Second, we prove Eq.~(\ref{eq:same-appendix}), where $ABC$ and $XYZ$ are shown in Fig.~\ref{fig:edge-appendix}(a). Importantly, $ABC$ is an annulus covering the entire edge and $XYZ$ is a bulk disk. For the bulk disk, the modular commutator $J(X,Y,Z)_{\vert \psi_{\rm 2D}\rangle}$ is invariant under deformations of the regions within the bulk~\cite{Kim2021,Kim2021a}. For the annulus $ABC$, we have proved above the analogous invariance under bulk-deformations of $ABC$. Therefore, for proving Eq.~\eqref{eq:same-appendix}, we can instead consider the subsystems $XYZ$ and $ABC$ shown in Fig.~\ref{fig:edge-appendix}(b) and (c). (Note that these subsystems are different from the ones shown in Fig.~\ref{fig:edge-appendix}(a).)
	
	Note that $I(A':Y|X)_{|\psi_{\rm 2D}\rangle}=I(C':Y|Z)_{|\psi_{\rm 2D}\rangle}=0$. This is because 
	\begin{equation}
		\begin{aligned}
			I(A':Y|X)_{|\psi_{\rm 2D}\rangle}  & =(S_{XY} + S_{YZW} -S_X - S_{ZW})_{|\psi_{\rm 2D}\rangle}-I(B'C':Y|ZW)_{|\psi_{\rm 2D}\rangle} \\
			&\leq (S_{XY} + S_{YZW} -S_X - S_{ZW})_{|\psi_{\rm 2D}\rangle}\\
			&= 0.
		\end{aligned}
		\label{eq:step1-v2}
	\end{equation}
	In the first line, we use the purity of the state; in the second line, the inequality follows from SSA. For deriving the third line, the important fact is that $XYZW$ is a disk away from the edge. Therefore, one can apply the bulk-axiom {\bf A1}. However, by SSA, the conditional mutual information is always non-negative, $I(A':Y|X)_{|\psi_{\rm 2D}\rangle}\geq 0$
	Thus, we conclude that $I(A':Y|X)_{|\psi_{\rm 2D}\rangle}=0$. A similar argument leads to $I(C':Y|Z)_{|\psi_{\rm 2D}\rangle}=0$.
	
	Using Lemma~\ref{lemma:deformation} and the purity of the global state, we find that:
	\begin{equation}\label{eq:JJJ-appendix}
		J(X,Y,Z)_{\vert \psi_{\rm 2D} \rangle }= J(A'X,Y,C'Z)_{\vert \psi_{\rm 2D} \rangle }= -J(A,B,C)_{\vert \psi_{\rm 2D} \rangle }.
	\end{equation}
	
	To prove the third statement of Proposition~\ref{prop:EB-1}, we observe that all the steps that lead to Eq.~\eqref{eq:JJJ-appendix} can still go through as long as the area law in the bulk disk $XYZW$ holds. Therefore, $J(X,Y,Z)_{\vert \widetilde{\psi}_{\rm 2D} \rangle } = -J(A,B,C)_{\vert \widetilde{\psi}_{\rm 2D} \rangle }$. Because $U_{\rm edge}$ is supported within $A'B'C'$, $| \widetilde{\psi}_{\rm 2D} \rangle$ and $| \psi_{\rm 2D} \rangle$ must have identical reduced density matrix on the disk $XYZW$.
	 Thus, $J(X,Y,Z)_{\vert \psi_{\rm 2D} \rangle } = J(X,Y,Z)_{\vert \widetilde{\psi}_{\rm 2D} \rangle } $. Therefore,  $J(A,B,C)_{\vert \psi_{\rm 2D} \rangle } = J(A,B,C)_{\vert \widetilde{\psi}_{\rm 2D} \rangle } $. This completes the proof of Eq.~(\ref{eq:U-appendix}). 
	
	Lastly, we prove the fourth statement of Proposition~\ref{prop:EB-1}. It is crucial to use axiom {\bf A0} here. Axioms {\bf A0} and {\bf A1} for disks of bounded radii (Fig.~\ref{fig:EB-setup-axioms}) imply that the same conditions hold on larger regions. Therefore, we can apply the enlarged version of {\bf A0} on the bulk disk $DE$ in Fig.~\ref{fig:edge-appendix}(d). The condition reads, 
	\begin{equation}
		(S_{DE} + S_D - S_E)_{| \psi_{\rm 2D} \rangle} =0 \quad \Longrightarrow \quad I(D:F|E)_{| \psi_{\rm 2D} \rangle} =0.
	\end{equation}
    A pure state with zero conditional mutual information has the following structure decomposition~\cite{Hayden2004}:
    \begin{equation}
    	| \psi_{\rm 2D} \rangle = | \psi_{DE_L} \rangle \otimes | \psi_{E_R F}\rangle.
    \end{equation}
    Here the labels $E_L$ and $E_R$ are associated with a decomposition of Hilbert space $\mathcal{H}_E = (\mathcal{H}_{E_L} \otimes \mathcal{H}_{E_R}) \oplus \cdots$. (Note, $E_L$ and $E_R$ are not labels for subsystems of the underlying lattice of the quantum system, unless the state is a product state.) Thus, any operator acting on $F$ will not change the reduced density matrix on $D$. In fact, $O_{F} | \psi_{\rm 2D} \rangle = \mu \, U_{EF} | \psi_{\rm 2D} \rangle$ for some real number $\mu$ and unitary operator $U_{EF}$. When $\mu \ne 0$, letting $\lambda=1/\mu$, we find $\lambda O_{\rm edge} \vert \psi_{\rm 2D}\rangle =U_{\rm edge} \vert \psi_{\rm 2D}\rangle$. Then statement 4 follows from statement 3. 
\end{proof}

\section{Angle conjecture in vacuum AdS$_{3}$ and the BTZ black hole}	

In order to properly formulate the ``angle conjecture'' of the main text, we need a theory of AdS$_3$/CFT$_2$ with a nonzero chiral central charge. For concreteness, we will consider the chiral gravity theory introduced in \cite{Li2008}. Every solution to Einstein's equations is also a solution of this theory. In particular, as explained in that reference, the vacuum state of a 1+1D, chiral, holographic CFT is dual to ordinary AdS$_3$ spacetime, but with modified boundary conditions that create a distinction between asymptotically right-moving and left-moving modes. Further, thermal states above the Hawking-Page temperature \cite{Hawking1982,Witten1998,Headrick2019} are still dual to the BTZ black holes of the non-chiral theory, with a similar modification in boundary conditions. It was shown in \cite{castro2014holographic} that entanglement entropies of intervals on the $t=0$ slice of these spacetimes are given by local quantities on the corresponding RT surfaces. So while the actual local quantity changes as compared to the non-chiral theory --- picking up a term that depends on a canonical normal frame associated to each geodesic, in addition to the non-chiral area term --- the ordinary geodesics are still the relevant bulk curves to consider in computing entanglement entropies.

In this Section, we demonstrate that for the AdS$_3$ vacuum, the formula Eq.~\eqref{eq:general_J_AdS} matches our main result \eqref{eq:main_result} for the infinite system $L=\infty,$ and the result \eqref{eq:eta-eff} for finite $L$ in the limit $\beta = \infty.$ In the static BTZ case, we demonstrate that Eq.~\eqref{eq:general_J_AdS} matches Eq.~\eqref{eq:eta-eff} in the limit $\beta \ll L.$ This is the right limit of Eq.~\eqref{eq:eta-eff} to apply, as the BTZ black hole is the bulk dual of a boundary thermal state only in the regime $\beta < L,$ since the Hawking-Page temperature is $T = 1/L.$

	\subsection{Global $\mathrm{AdS}_3$}
	\label{appendix:global_ads3_crossing}
	
			\begin{figure}[h]
			\centering
			\includegraphics[scale=1.05]{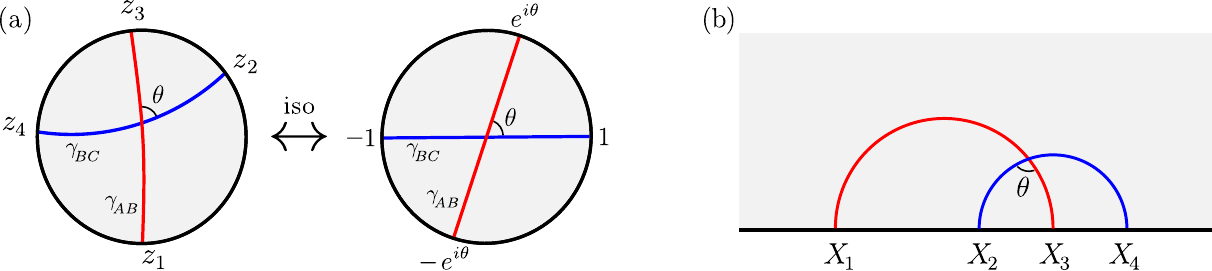}
			\caption{(a) The Poincar\'e disk model of hyperbolic plane $S$ and an isometry. (b) The upper half-plane model of the hyperbolic plane.}\label{fig:global-AdS}
			\end{figure}
			
	In the AdS$_3$/CFT$_2$ correspondence, the CFT vacuum state is dual to a bulk spacetime with the geometry of \textit{global AdS$_3$}. This is a spacetime with metric
	\begin{equation}\label{eq:global-AdS-coordinate}
		ds^2 = R^2(-\cosh^2 \rho \, dt^2 +d\rho^2 +\sinh^2 \rho \, d\phi^2),
	\end{equation}
	with $t\in (-\infty,\infty), \rho\in [0,\infty), \phi\in [0,2\pi).$ $R$ is the radius of curvature, and is related to the central charge $c := c_L + c_R$ by $c = 3 R / 2 G_N$~\cite{Brown1986}. Any constant-$t$ surface in AdS$_3$ is a \textit{moment of time symmetry}, meaning its extrinsic curvature vanishes. For any boundary region that lies entirely in a constant-$t$ slice of the boundary, the minimal geodesics homologous to that region lie entirely in the corresponding constant-$t$ slice of the bulk~\cite{engelhardt2014extremal}. Within such a slice, the induced metric is given by 
	\begin{equation}\label{eq:H2-metric-1}
		ds^2 = R^2(d\rho^2 +\sinh^2 \rho \, d\phi^2).
	\end{equation}
	
	The metric in Eq.~\eqref{eq:H2-metric-1} describes the hyperbolic plane $\mathbb{H}^2$. The hyperbolic plane is non-compact and therefore does not have a boundary. However, it has a ``conformal boundary,'' meaning that it is conformally equivalent to a metric that is extendible to a boundary curve at $\rho = \infty.$ This equivalence can be realized via the coordinate transformation $\rho \mapsto 2 \tanh^{-1}{r},$ which gives the ``Poincar\'{e} disk'' metric
	\begin{equation} \label{eq:metric_hyperbolic}
	    ds^2 = \frac{4 R^2}{1 - r^2} (dr^2 + r^2 d\phi^2).
	\end{equation}
	If we divide out by the global conformal factor $4 / (1-r^2),$ then the resulting metric can be extended smoothly to the boundary $r = 1,$ i.e., $\rho = \infty.$ This boundary has the geometry of a circle of radius $R.$ In the AdS$_3$/CFT$_2$ correspondence, the CFT vacuum state can be thought of as living on this circle.

   Geodesics of the Poincar\'{e} disk are circular arcs (or straight lines) perpendicular to the boundary $r=1$. Because the metric Eq.~\eqref{eq:metric_hyperbolic} is conformally related to the flat metric $dr^2 + r^2 d\phi^2,$ the apparent crossing angles in the plane are identical to the geometric crossing angles. In complex coordinates $z = r e^{i \phi},$ every isometry of the Poincar\'{e} disk is represented by a fractional linear transformation that maps the circle $|z| = 1$ to itself; these are expressed as
    \begin{equation}
    	z\rightarrow e^{i\alpha} \frac{z-\xi}{1-\bar{\xi}z},
    \end{equation}
    where $\alpha \in [0,2\pi)$ and $|\xi|<1$ is a complex number. These transformations form the isometry group $SL(2,R)$ of the Poincar\'{e} disk. As is well-known, such transformations leave complex cross-ratios $\eta_{\rm complex}\equiv {(z_{12}z_{34})}/{(z_{13}z_{24})}$ invariant, where $\{z_j\}$ are points on the complex plane. If $z_1$, $z_2$, $z_3$ and $z_4$ are points on the circle $|z|=1$, then the cross-ratio $\eta_{\rm complex}$ becomes real, and equals the cross-ratio defined in terms of the chord distance:
    \begin{equation}\label{eq:x-ratio-L}
    	\eta = \frac{\sin(\phi_{12}/2)\sin(\phi_{34}/2)}{\sin(\phi_{13}/2)\sin(\phi_{24}/2)}.
    \end{equation}
   
   Given any two intersecting geodesics in $\mathbb{H}^2,$ one can apply an isometry to send the intersection point to the center of the disk. The two geodesics become straight lines and the four boundary points are mapped to two pairs of antipodal points; see Fig.~\ref{fig:global-AdS}(a). Then it follows from a simple trigonometry calculation that
   \begin{eqnarray}\label{eq:angle_AdS}
   	\cos\theta = 2\eta -1, \quad  \textrm{where $\eta$ is given by Eq.~\eqref{eq:x-ratio-L}.}
   \end{eqnarray}
   With the identification $\phi = 2 \pi x / L,$ this matches the finite-$L$ vacuum formula given in the main text in Eq.~\eqref{eq:eta-eff}.\footnote{In fact, this identity can be verified explicitly using the coordinates \eqref{eq:H2-metric-1} as well, making use of the explicit geodesic solution
   	\begin{equation}
   		\gamma_{AB}:~\tanh \rho \cos \left(\phi-\frac{\phi_1+\phi_3}{2}\right) = \cos \left(\frac{\phi_{13}}{2}\right).
   	\end{equation}
   	}
   
   To study the $L=\infty$ realization of the AdS$_3$/CFT$_2$ duality, we make the coordinate transformation
   \begin{equation}
   	z = \frac{w-i}{w+i},
   \end{equation}
   where $w=X+iY$ with $Y>0$. This maps the hyperbolic disk to the upper half-plane (see Fig.~\ref{fig:AdS_appendix}(c)). In planar coordinates $(X,Y)$, the metric becomes
   \begin{equation}
   	ds^2 = \frac{R^2}{Y^2}(dX^2+dY^2).
   \end{equation}
   In these coordinates, geodesics are semicircles or straight lines perpendicular to the boundary at $\{Y=0\} \cup \{\infty\}.$ One can easily verify
   \begin{equation}\label{eq:x-ratio-upper}
   \cos\theta =	2\eta -1,   \quad  \textrm{where}\quad   \eta = \frac{X_{12} X_{34}}{X_{13} X_{24}}.
   \end{equation}
Here $X_j$ are the $X$ coordinates of four points on the boundary, with $X_{ij} \equiv X_j-X_i$. This matches the main result Eq.~\eqref{eq:main_result}.
	
	\subsection{BTZ black hole}
	\label{appendix:crossing_btz}
	
		BTZ black holes are a class of solutions in 2+1D Einstein gravity with a negative cosmological constant~\cite{BTZ1992,BHTZ1993}. These solutions are labeled by the black hole mass $M$ ($M>0$) and the angular momentum $J$.\footnote{As explained in \cite{Li2008}, the asymptotically measured mass and angular momenta differ in a chirally modified theory from the ``bare'' quantities appearing in the metric.} Importantly, the solutions can be realized as quotients of global $\mathrm{AdS}_3$~\cite{Carlip1995}.  We shall be interested in the non-rotating solutions with $J=0.$ These solutions are expressed in ``Schwarzschild coordinates'' as
			\begin{equation}\label{eq:Schwarzschild-coordinate}
			ds^2 = -\frac{(r^2-r^2_{+})}{R^2} dt^2+\frac{R^2}{r^2-r^2_{+}} dr^2 +r^2 d\phi^2,\qquad t \in (-\infty, +\infty),\,\, r \in (r_+ , +\infty), \,\, \phi \in [0, 2\pi),
		\end{equation}
	    where $r_+ = \sqrt{M} R$ is the  black hole horizon. The singularity at $r=r_+$ is a coordinate singularity, and the metric can be analytically extended past this point to the ``black hole interior.'' The inverse temperature of the black hole is given by $\beta = 2 \pi R^2 / r_+ = L/\sqrt{M}.$

We shall be interested in the $t=0$ slice of the BTZ exterior.  The induced metric on this slice is $ds^2= \frac{R^2}{r^2-r^2_{+}} dr^2 +r^2 d\phi^2$. The coordinate transformation
\begin{equation}
	    		X= \frac{\sqrt{r^2-r_+^2}}{r}   \exp\left( \sqrt{M} \phi \right) \quad \textrm{and} \quad
	    		Y=  \frac{r_+}{r} \exp\left(\sqrt{M} \phi\right)
\end{equation}
identifies this single-time slice with the subset $1 \leq X^2 + Y^2 \leq \exp(4 \pi \sqrt{M}), X\geq 0$ of the upper half-plane with metric
\begin{equation}
	ds^2 = \frac{R^2}{Y^2} ( dX^2 +dY^2).
\end{equation}
In this mapping, the curves $X^2 + Y^2 = 1$ and $X^2 + Y^2 = e^{4 \pi \sqrt{M}}$ are identified. In fact, by taking $\phi$ outside of the range $[0, 2 \pi),$ we see that the $X \geq 0$ region of the upper half-plane, minus the point $X=Y=0,$ ``wraps around'' the BTZ exterior infinitely many times. Each strip $\exp(4 \pi n \sqrt{M}) \leq X^2 + Y^2 \leq \exp(4 \pi (n+1) \sqrt{M})$ corresponds to a single copy of the BTZ exterior. This is the statement that the BTZ black hole can be realized as a quotient of AdS$_3.$

We may specify three contiguous boundary intervals in these coordinates by choosing consecutive points $X_1, X_2, X_3, X_4$ in the interval $[1, \exp(2 \pi \sqrt{M})].$ Geodesics anchored on these boundary intervals are represented in these coordinates as circular arcs perpendicular to the boundary. However, there are two geodesics anchored to any two points $X_{j}$ and $X_{k}$; one that is a continuous segment lying entirely in the fundamental domain $1 \leq X^2 + Y^2 \leq \exp(2 \pi \sqrt{M}),$ and there is one that ``wraps around the quotient'' by going through the identified bulk curves; see Fig.~\ref{fig:AdS_appendix}. While these ``wrapping'' curves are not homologous to the corresponding boundary intervals, the union of one of these ``wrapping'' curves with the BTZ horizon $X = 0, Y \in [1, \exp(2 \pi \sqrt{M})]$ \textit{is} homologous. Whether the non-wrapping geodesic, or the wrapping geodesic plus the horizon, has minimal area and is thus the RT surface of the corresponding interval, depends on the size of that interval. Above a certain critical angular size $\Delta \phi_c$, the minimal geodesic of an interval is the wrapping one plus the horizon; below this angular size, the minimal geodesic is the non-wrapping one. Explicit calculation of the various geodesic lengths shows that the critical angular size satisfies
\begin{equation}
    \frac{\sinh(\sqrt{M} \Delta \phi_c/2)}{\sinh(\sqrt{M} (2 \pi - \Delta \phi_c)/2)}
        = \exp(\sqrt{M} \pi).
\end{equation}
In the large-$M$ limit, this asymptotes to
\begin{equation}
    \Delta \phi_c \sim 2 \pi - \frac{\log{2}}{\sqrt{M}}.
\end{equation}

\begin{figure}[h]
	\centering
	\subfloat[BTZ black hole]{\includegraphics[width=0.18\linewidth]{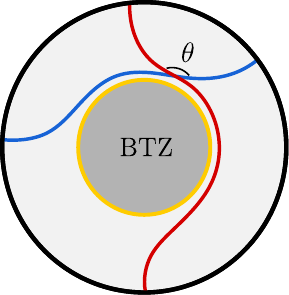}}
	\hspace{0.8cm}
	\subfloat[BTZ black hole as quotient of AdS]{\includegraphics[width=0.35\linewidth]{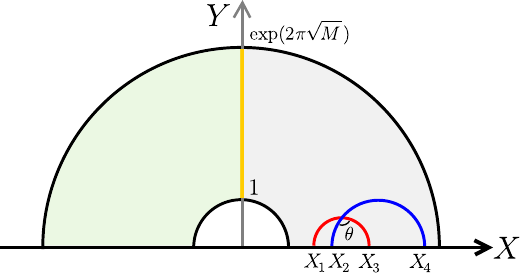}}
	\hspace{0.8cm}
	\subfloat[a ``wrapping" geodesic]{\includegraphics[width=0.35\linewidth]{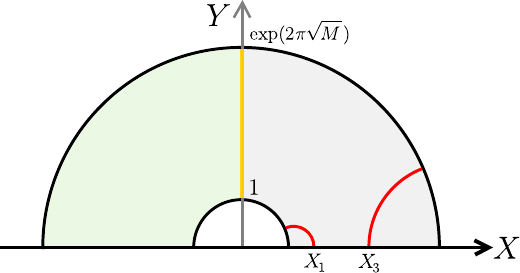}}
	\caption{
		A single time-slice of the BTZ black hole geometry outside the horizon. (a) The red and blue lines represent geodesics. The yellow circle represents the horizon. The boundary at infinity is the black circle. Unlike the Poincar\'e disk, the apparent crossing angle here does not faithfully  represent the geometric crossing angle. (b) Embedding into the upper half-plane. The inner semicircle and the outer semicircle are identified. The green region on the left of the horizon (yellow) is the other side of the two-sided Kruskal-like extension of the BTZ black hole (outside the horizon). The red and blue lines are the ``non-wrapping'' geodesics for the intervals $[X_1, X_3]$ and $[X_2, X_4].$ (c) An example of the ``wrapping'' geodesic for the interval $[X_1, X_3],$ which is the minimal one when its length, plus the horizon length, is smaller than the length of the non-wrapping geodesic.}
	\label{fig:AdS_appendix}
\end{figure}	

To completely test the conjecture Eq.~\eqref{eq:general_J_AdS}, we would need to test the cases: (i) neither geodesic wraps the quotient, (ii) one geodesic wraps the quotient, (iii) both geodesics wrap the quotient. However, the limit under which we have analytical control over the modular commutator in a thermal state via Eq.~\eqref{eq:eta-eff} is $\sqrt{M} = L/\beta \rightarrow \infty.$ In that limit, we see that the critical angular distance becomes $\Delta \phi_c = 2 \pi.$ So to check that the angle conjecture reproduces the analytic formula in part two of Eq.~\eqref{eq:eta-eff}, we need only check case (i).

In this case, which is shown in Fig.~\ref{fig:AdS_appendix}(a),(b), the crossing angle can be computed using the parameters in the upper half-plane model in the same manner as Eq.~\eqref{eq:x-ratio-upper}.
We write
\begin{equation}\label{eq:idea-effective-1}
	\cos \theta = 2\eta_{\rm eff} -1, \qquad \textrm{where}\quad \eta_{\rm eff} = \frac{X_{12}X_{34}}{X_{13} X_{24}}.
\end{equation}
To write $\eta_{\rm eff}$ in terms of the parameters $x_i \in [0, L)$ of the boundary circle, we rewrite
\begin{equation}\label{eq:idea-effective-2}
	X_j = \exp \left(\frac{2\pi \sqrt{M} x_j}{L}\right) \quad \Rightarrow \quad \eta_{\rm eff} = \frac{\sinh\left(\pi \sqrt{M} x_{12}/L\right)\sinh\left(\pi \sqrt{M} x_{34}/L\right)}{\sinh\left(\pi \sqrt{M} x_{13}/L\right)\sinh\left(\pi \sqrt{M} x_{24}/L\right)}.
\end{equation}
This matches part two of Eq.~\eqref{eq:eta-eff} once we identify $\sqrt{M} = L / \beta$.

\section{Beyond contiguous intervals}

	\begin{figure}[h]
	\centering
	\includegraphics[scale=1.05]{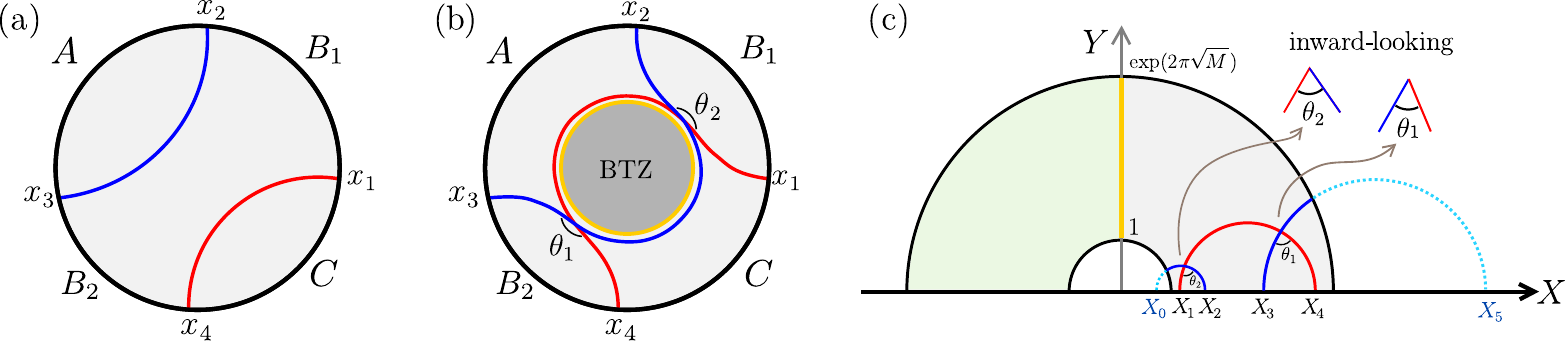}
	\caption{The boundary circle is partitioned into $ABC$ where $B$ is the union of disjoint intervals $B_1$ and $B_2$. (a) In pure AdS, the two geodesics ($\gamma_{AB}$ and $\gamma_{BC}$) have no intersection. (b) For the BTZ black hole geometry in the limit $\beta \ll L$, the two geodesics intersect at two places.  (c) Embedding of the region in the upper half-plane. The angles $\theta_1$ and $\theta_2$ are represented accurately. They can be calculated from $\{X_i\}_{i=0}^5$. }
	\label{fig:AdS_disjoint}
\end{figure}

	We conjecture that the contributions from each crossing angle are additive. Specifically, suppose the boundary CFT is dual to a semiclassical bulk geometry with a moment of time symmetry and the RT surfaces intersect at a set of crossing angles $\{\theta_i\}$. We conjecture that
	\begin{equation}\label{eq:general_J_AdS-appendix}
		\boxed{
			J(A,B,C) = \frac{ \pi c_-}{6} \sum_i s_i \cos\theta_i.
		}
	\end{equation}
    Here $s_i= \pm 1$ are introduced to fix the sign convention. Let us explain the rule for assigning $s_i$ to angle $\theta_i$. If the angle $\theta_i$ has geodesic $\gamma_{AB}$ on its right from an inward-looking point of view, we assign $s_i=1$; otherwise, $s_i=-1$.  
    
    Let us make a few side remarks. First, we only need to pick one crossing angle (out of the four) at each intersection point. It does not matter which angle we pick, because different choices are related by the sign convention $s_i$. This prescription automatically ensures $J(A,B,C) = -J(C,B,A)$, which is a property that must hold generally for the modular commutator.  Secondly, one may ask what happens if the bulk spacetime does not possess a moment of time symmetry, and so the geodesics do not intersect. This can happen for more interesting choices of intervals of the boundary theory, such as those on a general Cauchy surface discussed in Section~\ref{sec:Cauchy}. We believe that the conjecture can be generalized to that setup, and we leave this to future investigation. 

As a warmup, let us first consider the pure $\mathrm{AdS}_3$ case with the subsystems $A$, $B$, and $C$ covering the entire boundary. We choose $B$ to be a union of two disjoint intervals, denoted as $B_1$ and $B_2$; see Fig.~\ref{fig:AdS_disjoint}(a). In the Poincar\'e disk model, it is evident that the two geodesics ending on disjoint boundary intervals ($A$ and $C$) cannot intersect. Furthermore, geodesics $\gamma_{AB}$ and $\gamma_{C}$ are identical since, on a (global) hyperbolic plane, two points determine a unique geodesic. Similarly, $\gamma_A$ is identical to $\gamma_{BC}$ for the same reason. Therefore, $\gamma_{AB}$ and $\gamma_{BC}$ do not intersect. Thus, according to our holographic conjecture 
\begin{equation}\label{eq:ABBC-fact1}
	\boxed{
	J(A,B,C)_{|\Omega\rangle}=0
	}
\end{equation}
for the CFT vacuum $|\Omega\rangle$. Note that this holds for any CFT, even the chiral ones. This fact can be independently verified by the operator-based method of Section~\ref{appendix:operators} as follows. Suppose the boundary circle has a circumference of $L$ and $0<x_1<x_2<x_3<x_4< L$. We can derive Eq.~\eqref{eq:ABBC-fact1} using the Cardy-Tonni expression of the modular commutator~\cite{cardy_entanglement_2016}. Let $B_1=[x_1,x_2]$ and $B_2=[x_3,x_4]$ with $0\leq x_1<x_2<x_3<x_4< L$, and $A=[x_2,x_3]$ and $C=[x_4,L+x_1]$, then 
\begin{equation}
	K_{AB} = \frac{L}{\pi}\int_{x_1}^{x_4} dx\, \frac{\sin(\pi (x-x_1)/L)\sin(\pi(x_4-x)/L)}{\sin(\pi x_{14}/L)}(T(x)+\bar{T}(x))
\end{equation}
and
\begin{equation}
	K_{BC} = \frac{L}{\pi}\int_{x_3}^{L+x_2} dx\, \frac{\sin(\pi (x-x_3)/L)\sin(\pi(L+x_2-x)/L)}{\sin(\pi (L-x_{23})/L)}(T(x)+\bar{T}(x))
\end{equation}
Computing the commutator using Eqs.~\eqref{eq:uu-commutator} and \eqref{eq:vv-commutator} and identifying $T(x)$ with $T(x+L)$ for $x\in [x_1,x_2]$ we find $J(A,B,C)$ identically vanishes.

Now we discuss a case with a pair of crossing angles $\{\theta_1,\theta_2\}$. This is similar to the setup discussed above, except now we have a BTZ black hole; see Fig.~\ref{fig:AdS_disjoint}(b) and (c). As in the previous Section, we will restrict our attention to the limit $\beta \ll L$ so we may always consider the ``non-wrapping'' RT surfaces instead of the ``wrapping'' ones. We shall derive an expression for the modular commutator using our proposed correspondence; this result shall be compared against a numerical experiment, discussed in Section~\ref{appendix:numerical_test_holo}.

Let us first derive the expressions for the crossing angles on the AdS side. By mapping the BTZ exterior to a portion of the upper half-plane as in Section~\ref{appendix:crossing_btz}, for black holes in the limit $\beta \ll L$, one can verify 
\begin{equation}
	\label{eq:BTZ_angle_eff}
		\cos \theta_1 = 2\frac{X_{13}X_{45}}{X_{14}X_{35}}-1 \quad \textrm{and} \quad 
		\cos \theta_2 = 2\frac{X_{01}X_{24}}{X_{02}X_{14}}-1.
\end{equation}
Here $X_0$ and $X_5$ are the images of $X_4$ and $X_2$ in the neighboring preimages of the quotient; see Fig.~\ref{fig:AdS_disjoint}(c). The fractions $({X_{13}X_{45}})/({X_{14}X_{35}})$ and $({X_{01}X_{24}})/({X_{02}X_{14}})$ can be interpreted as effective cross ratios; see Eq.~\eqref{eq:idea-effective-1} for a comparison. The coordinates $\{X_i\}_{i=0}^5$ are the endpoints on the $X$ axis, and they are related to the coordinates on the boundary circle $\{x_j\}_{j=1}^4$ according to the relation
$	X_j = \exp( {2\pi \sqrt{M}x_j}/{L})$ for $j=1,2,3,4 $, $X_0= \exp(-2\pi \sqrt{M})X_3$ and $X_5= \exp(2\pi \sqrt{M})X_2$. In other words, $\cos\theta_1$ and $\cos\theta_2$ are functions of the boundary coordinates $\{x_j\}_{j=1}^4$ and $\sqrt{M}/L$.

Now we relate the crossing angles to the CFT data. In the high-temperature regime $\beta \ll L,$ i.e. $M \gg 1$, our proposal for the thermal state $\rho^{\beta}$ reads:
\begin{equation}
	\label{eq:J2_BTZ}
	\boxed{
	J(A,B,C)_{\rho^{\beta}} = \frac{\pi c_{-}}{6} (\cos \theta_1 - \cos \theta_2).
	}
\end{equation}
A minus sign appears in the second term due to the rule we specified above. The two angles are related to the effective cross-ratio
\begin{equation}
\label{eq:eta_eff_appendix}
	\eta^{\beta}_{\text{eff}}(a, b, c) \equiv \frac{\sinh(\pi a/\beta)\sinh(\pi c/\beta)}{\sinh(\pi (a+b)/\beta)\sinh(\pi (b+c)/\beta)}
\end{equation}
according to
\begin{equation}
	\label{eq:BTZ_angle_2}
	\begin{aligned}
		\cos \theta_1 &= 2\eta^{\beta}_{\text{eff}}(|AB_1|, |B_2|, |B_1C|)-1, \\
		\cos \theta_2 &= 2\eta^{\beta}_{\text{eff}}(|B_2C|, |B_1|, |AB_2|)-1,
	\end{aligned}
\end{equation}
where $|A|$ is the length of the boundary interval $A$. As a simple consequence, Eq.~\eqref{eq:eta_eff_appendix} gives nonzero result only if $x_{12} \ne x_{34}$, that is $|B_1| \ne |B_2|$.

	\subsection{Numerical test of Eq.~\eqref{eq:J2_BTZ}}
	\label{appendix:numerical_test_holo}

	Eq.~\eqref{eq:J2_BTZ} can be tested numerically for chiral thermal states of a free fermion CFT, which is a nonchiral 1+1D CFT with central charge $c=1/2$ ($c=c_L=c_R$). Numerical techniques are discussed in Section~\ref{appendix:ffnumerics}.
	 Consider a chiral thermal state $\rho^{(\beta_L,\beta_R;L)}$, i.e., a chiral thermal state with inverse temperature $\beta_L$ ($\beta_R$) for the left (right) moving modes on a circle with length $L$.  Consider the chiral thermal state with $\beta_R \to + \infty$ and consider the $ABC$ partition in Fig.~\ref{fig:AdS_disjoint}. It follows directly from the discussion in the previous Section that
	\begin{equation}\label{eq:chiral-thermal-2thetas}
	J(A,B,C)_{\rho^{(\beta_L,+\infty;L)}} = \frac{\pi }{6} \left( 
	\frac{\sinh(\pi |AB_1|/\beta)\sinh(\pi |B_1C|/\beta)-\sinh(\pi |AB_2|/\beta)\sinh(\pi |B_2C|/\beta)}{\sinh(\pi |AB|/\beta)\sinh(\pi |BC|/\beta)}
	\right).
	\end{equation}
   
   	\begin{figure}[h]
		\centering
		\subfloat[]{\includegraphics[width=0.40\linewidth]{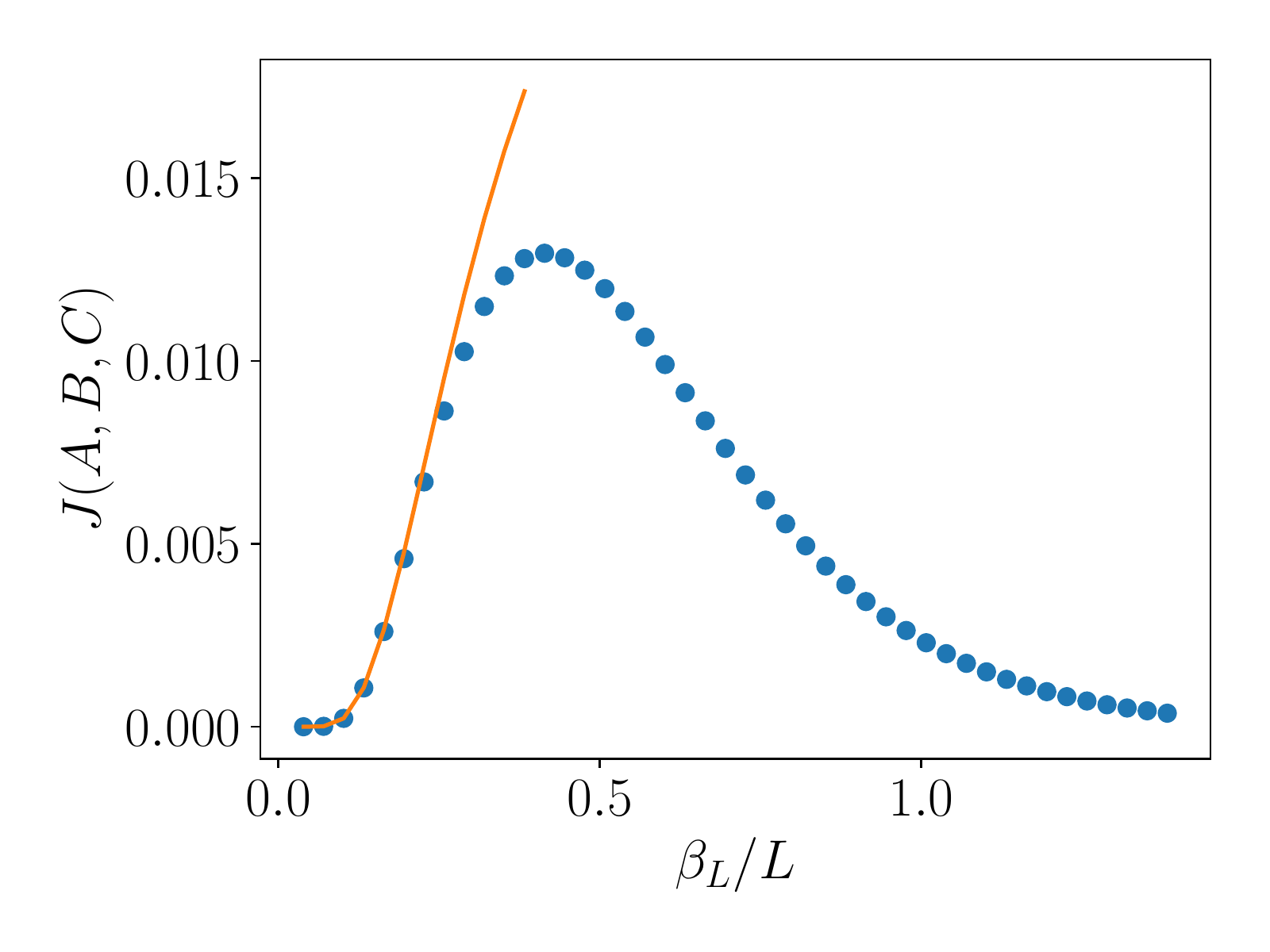}}
		\hspace{0.2cm}
		\subfloat[]{\includegraphics[width=0.40\linewidth]{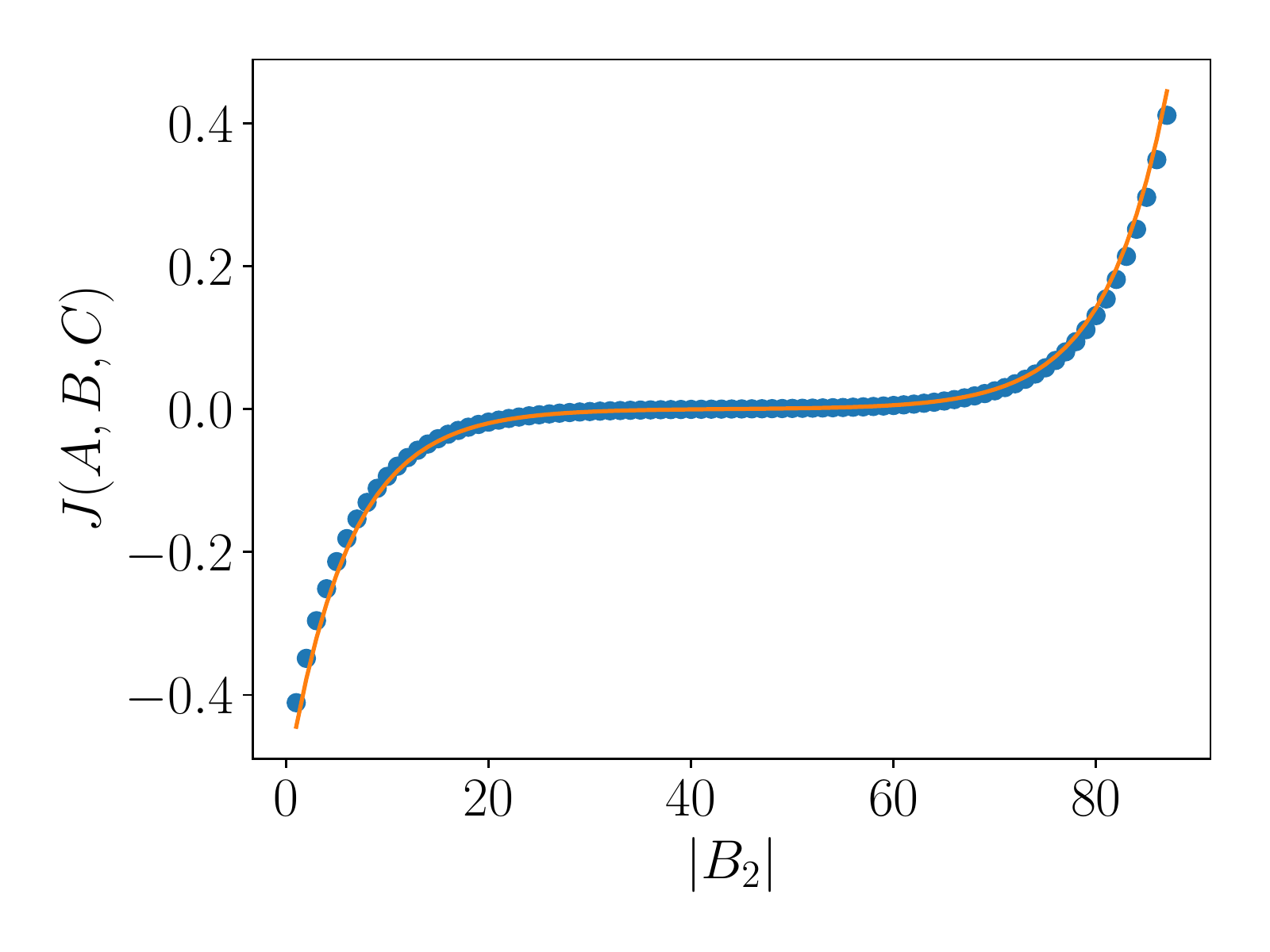}}
		\caption{Modular commutator $J(A,B,C)$ of a chiral thermal state $\rho^{(\beta_L,\beta_R,L)}$ of the free Majorana fermion CFT, where $B$ is composed of disjoint intervals $B_1$ and $B_2$. We fix $\beta_R=\infty$ and $L=128$ throughout the simulation. Blue dots represent numerical data and the orange line represents fitting from Eq.~\eqref{eq:chiral-thermal-2thetas}. (a) We fix the lengths of intervals $|A|=|C|=L/4, |B_1| = L/4-2, |B_2|=L/4+2$, and varying $\beta_L$. (b) We fix $\beta_L=0.15L$ and $|A|=|C|=20$ and vary $|B_2|$. }\label{fig:AdS_disjoint-text}
	\end{figure}
   
   First, a consequence of Eq.~\eqref{eq:chiral-thermal-2thetas} is that $J(A,B,C)_{\rho^{(\beta_L,+\infty;L)}}=0$ when $x_{12}=x_{34}$ (i.e., $|B_1|=|B_2|$).
	We can therefore make slight perturbation around this ``symmetric" configuration and see how $J$ changes. 

	Eq.~\eqref{eq:chiral-thermal-2thetas} is confirmed numerically for $\beta_L \ll L$ in the chiral thermal state of the free fermion CFT; see Fig.~\ref{fig:AdS_disjoint-text}(a). We see from this figure that when $\beta_L/L$ is small ($\beta_L/L<0.2$) the prediction fits the data well. 
	
	Second, if we keep $\beta_L = 0.15 L$ and vary the sizes of the intervals, our conjecture agrees with the data in good precision. This is illustrated in Fig.~\ref{fig:AdS_disjoint-text}(b).
	We note however that if the size of $|A|$ or $|C|$ are small compared with $\beta_L$, the data can deviate from our formula. The correction may come from three sources. (1) The free Majorana fermion is not a holographic CFT,
	(2) Even for holographic CFTs, for many boundary intervals there is a critical value of $\beta/L$ at which that interval enters the ``wrapping'' phase of its RT surface, changing the intersection angles appearing in our conjecture for the modular commutator, (3) Holographic CFTs are geometric only in the limit of large central charge, and finite-central-charge corrections may be important. We leave the study of these corrections for future work.

	\section{Free fermion numerical methods and simulation results}
	\label{appendix:ffnumerics}
	We discuss the numerical method we use for the simulation of free fermion systems. The main simplification arising in this setup is that both the ground states and thermal states are Gaussian. Therefore, the state is determined by the correlation matrix~\cite{Peschel2009}. 
	
	\subsection{Free fermion simulation of the Chern insulator}
	Consider first the $U(1)$ symmetric case with a chain of fermionic degrees of freedom. The creation and annihilation operators are denoted $c^{\dagger}_{i},c_i$, where $i=1,2\cdots N$ and $\{c^{\dagger}_{i},c_j\} = \delta_{ij}$. A Gaussian state $\rho$ is completely determined by the $N\times N$ correlation matrix
	\begin{equation}
		C_{ij} = \Tr(\rho c^{\dagger}_{i} c_j)
	\end{equation}
	The modular Hamiltonian $K=-\ln \rho$ is of the fermionic bilinear form
	\begin{equation}
		K = \sum_{ij} K_{ij} c^{\dagger}_i c_j.
	\end{equation}
	The matrix $K$ is related to the correlation matrix by
	\begin{equation}
		K = \ln \frac{I-C}{C}.
	\end{equation}
	The formula is applicable to any subsystem. Let $C_{AB}$ be the submatrix restricted to the fermions in region $AB$, then
	\begin{equation}
		K_{AB} = \ln \frac{I-C_{AB}}{C_{AB}}.
	\end{equation}
	The modular commutator can then by computed by
	\begin{equation}
		J(A,B,C) = i\Tr(C_{ABC} [K_{AB},K_{BC}])
	\end{equation}
	The Chern insulator can be modeled by the Hofstadter model \cite{Hofstadter1976,TKNN:82,Niu1985} in a two-dimensional square lattice,
	\begin{equation}
		H= -t\sum_{\vec{x},\vec{a}} (c^{\dagger}_{\vec{x}} e^{-i \vec{a} \cdot \vec{A}(\vec{x})} c_{\vec{x}+\vec{a}} + h.c.)+\mu \sum_{\vec{x}} c^{\dagger}_{\vec{x}} c_{\vec{x}}
	\end{equation}
	where $\vec{a}$ runs over lattice vectors and $\vec{A}$ is the vector potential which equals $\vec{A} = (0,B x_1)$ for the square lattice in the Landau gauge and $B$ is the flux per unit cell. We choose $\mu/t =-2.0$ and $B=\pi/2$ such that the lowest band is filled. The band has Chern number $1$, which gives the chiral central charge $c_-=1$. We put the system on an open cylinder with horizontal direction compactified on a circle with circumference $N$ and height $M$. We compute the modular commutator $J(A,B,C)$, where $A,B,C$ are rectangular strips on the boundary with the horizontal length $L_A,L_B,L_C$ and the same width $w$, where all of $L_A,L_B,L_C,w,M-w$ are sufficiently large compared to the correlation length. 
	
	The chiral edge mode is robust against small disorder. One may add an Anderson term~\cite{Anderson1958},
	\begin{equation}
		H= -t\sum_{\vec{x},\vec{a}} (c^{\dagger}_{\vec{x}} e^{-i \vec{a} \cdot \vec{A}(\vec{x})} c_{\vec{x}+\vec{a}} + h.c.)- \mu  \sum_{\vec{x}} c^{\dagger}_{\vec{x}} c_{\vec{x}} + \sum_{\vec{x}} V_{\vec{x}}c^{\dagger}_{\vec{x}} c_{\vec{x}},
	\end{equation}
	where $V_{\vec{x}}$'s are independent random variables drawn from the uniform distribution in $[-W/2,W/2]$, and $W$ is the disorder strength. The original Landau band gets broadened into a disordered band with width $O(W)$ that consists of localized eigenstates on the edge of the band and extended states in the middle of the band. As the chemical potential gets varied, the system goes through an Anderson localization transition and the Hall conductance jumps from $0$ to $1$ (in the unit of $e^2/h$). It is believed that the Anderson localization transition is responsible for the plateau of Hall conductance observed in experiments~\cite{Klitzing1980}, although the nature of the transition is still under debate \cite{Zhu2019,Sbierski2021}. 
	
	Since the modular commutator detects the chiral central charge, we expect there to be a similar transition in the modular commutator. Let $A,B,C$ be equal-size strips that cover the whole boundary, which gives $J(A,B,C) = \pi c/3$, it is expected then that the disordered averaged $J$ has a plateau transition as one varies the chemical potential. This is indeed observed, see Fig.~\ref{fig:J_2D_appendix}.  This also provides numerical evidence of the claim that $J$ is robust against disorder, as argued using entanglement bootstrap in the main text.
	\begin{figure}
		\centering
		\includegraphics[width=0.44\linewidth]{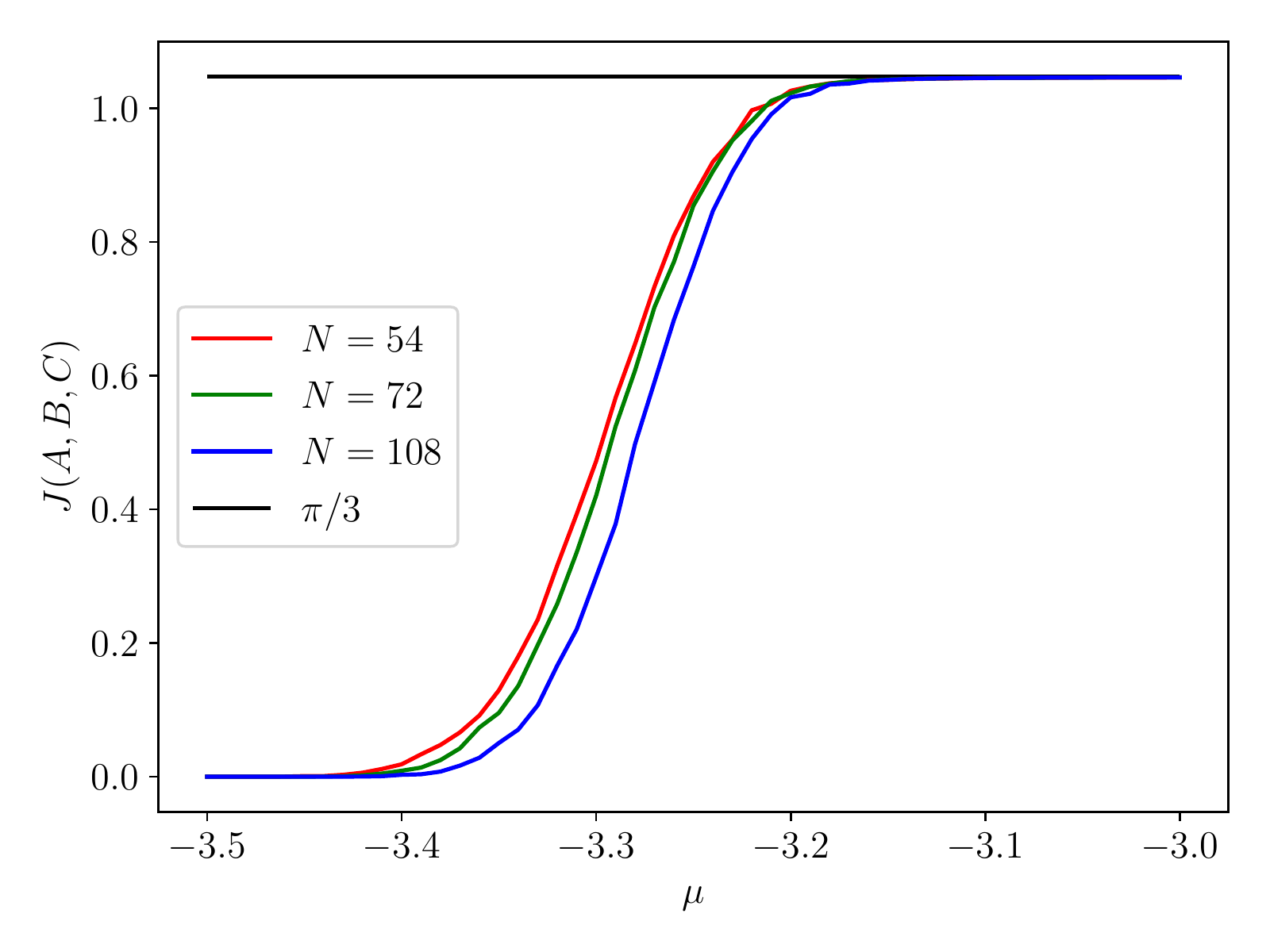}
		\includegraphics[width=0.44\linewidth]{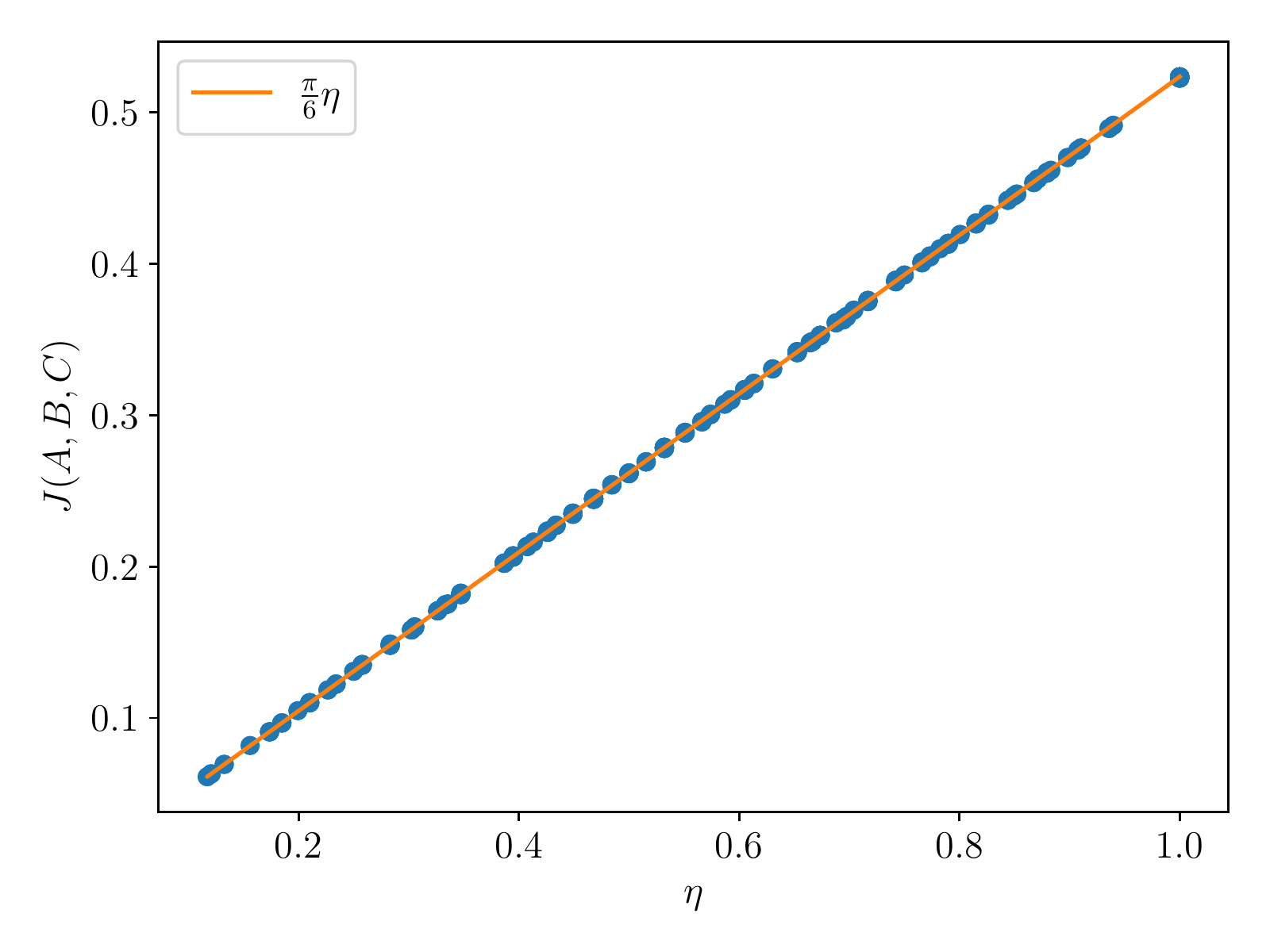}
		\caption{Left: Modular commutator $J(A,B,C)$ versus chemical potential $\mu$ in the disordered Chern insulator. We choose $B=2\pi/9$ and the width of the strips to be $w=6$, and $A,B,C$ cover the entire circle. Disordered average is performed on 300 samples. Right: Modular commutator $J(A,B,C)$ versus $\eta$ in $p+ip$ topological superconductor.}
		\label{fig:J_2D_appendix}
	\end{figure}
	\subsection{Free fermion simulation of the $p+ip$ topological superconductor}
	In this Section we consider the real fermions (Majorana fermions) with $\mathbb{Z}_2$ fermionic parity symmetry. Let $\psi_{i}, ~i=1,2\cdots 2N$ be the Majorana operators with anticommutation relations $\{\psi_i,\psi_j\} = 2\delta_{ij}$, a Gaussian state $\rho$ is completely specified by the correlation matrix, which is a real skew-symmetric $2N\times 2N$ matrix
	\begin{equation}
		M_{ij} = \Tr(-i(\psi_{i}\psi_{j}-\delta_{ij})\rho).
	\end{equation}
	It can be block diagoalized by an orthogonal matrix $O\in O(2N)$,
	\begin{equation}
		M = O (\text{diag}(n_i) \otimes (i\sigma^y)) O^{T},
	\end{equation}
	where $\sigma^y$ is the Pauli operator, and $n_i\in [-1,1]$. Similar to the complex fermion case one can write the modular Hamiltonian as
	\begin{equation}
		K = \frac{i}{2}\sum_{ij}K_{ij} \psi_i \psi_j,
	\end{equation}
	where $K_{ij}$ is a real skew-symmetric matrix
	\begin{equation}
		K = O \left(\text{diag}\left(\frac{1}{2}\ln \frac{1-n_i}{1+n_i}\right) \otimes (i\sigma^y) \right) O^{T},
	\end{equation}
	One may obtain the modular Hamiltonian of any subsystem by block diagonalizing submatrices of $M$. The modular commutator can be computed by
	\begin{equation}
		J(A,B,C) = \Tr([K_{AB},K_{BC}] M_{ABC} ).
	\end{equation}
	As a concrete application, we consider the BdG $p$-wave superconductor Hamiltonian \cite{Read2000,Bernevig_2013}
	\begin{align}
		H = \sum_{m,n} &-t(c^{\dagger}_{m+1,n} c_{m,n} + c^{\dagger}_{m,n+1} c_{m,n} + h.c.)-(\mu-4t)c^{\dagger}_{m,n}c_{m,n} \\
		&+ (\Delta c^{\dagger}_{m+1,n}c^{\dagger}_{m,n} +\Delta^{*} c_{m,n}c_{m+1,n}) + (i\Delta c^{\dagger}_{m,n+1}c^{\dagger}_{m,n} -i\Delta^{*} c_{m,n}c_{m,n+1})
	\end{align}
	of $L_x L_y$ complex fermions. 
	This can be recast into a Majorana fermion Hamiltonian with $2L_x L_y$ Majorana fermions using
	\begin{equation}
		\psi_{m,n,1} = c^{\dagger}_{m,n} + c_{m,n},~~ \psi_{m,n,2} = i(c^{\dagger}_{m,n} - c_{m,n}).
	\end{equation}
	We take the boundary condition in the $x$ direction to be Neveu-Schwarz (NS), and $y$ direction to be open. In this setting there is no Majorana zero mode and the ground state is unique. We choose $t=1, \Delta = 0.5$ and $\mu=1$ such that the system is in the topologically nontrivial phase. On the boundary there is a chiral Majorana fermion CFT in the NS sector, with $c=1/2$. In the actual finite-size simulations we choose $L_x=36$ and $L_y=20$ and the width of strips $A,B,C$ to be $w=8$; these strips are placed near the lower boundary and we varied the horizontal lengths of the strips. We find perfect agreement with the analytical result $J(A,B,C) = \pi c \eta/3$, see Fig.~\ref{fig:J_2D_appendix}.

	\subsection{Chiral thermal state on the lattice}
	We can also construct a chiral thermal states for free fermion systems directly on a $1+1$D lattice without going to $2+1$D. Consider a one-dimensional complex fermion chain in infinite space with Hamiltonian
	\begin{equation}
		H = -t \sum_{i} (c^{\dagger}_i c_{i+1} + h.c.)
	\end{equation}
	The Hamiltonian can be diagonalized using Fourier modes $c(k) = \sum_j c_j e^{-ikj} $,
	\begin{equation}
		H = \int_{-\pi}^{\pi} dk\, \epsilon(k) c^{\dagger}(k) c(k),
	\end{equation}
	where 
	\begin{equation}
		\epsilon(k) = -2t \cos k.
	\end{equation}
	The momenta near the Fermi points $k=\pm \pi/2$ determine the low-energy CFT, which is a $1+1$D free Dirac fermion. We will choose $t=1/2$ to normalize the Hamiltonian such that the speed of light is one. A quantum state is completely specified by the expectation value of $n(k) = c^{\dagger}(k) c(k)$, where for the ground state we have $n(k) = \Theta(-\epsilon(k))$ and $\Theta(x)$ is the step function. For a chiral thermal state with inverse temperatures $(\beta_L,\beta_R)$, we expect two Fermi-Dirac distributions depending on the chirality, 
	\begin{equation}
		n(k) = \left\{ \begin{array}{ll}
			1/({1+e^{\beta_L \epsilon(k)}}) &\mbox{ if $k<0$}, \\
			1/({1+e^{\beta_R \epsilon(k)}}) &\mbox{ if $k>0$}
		\end{array} \right.
	\end{equation}
	This completely determines the state. In particular, we can write down the correlation matrix
	\begin{equation}
		C_{mn} = \frac{1}{2\pi}\int_{-\pi}^{\pi} dk\, n(k) e^{ik(n-m)},
	\end{equation}
	from which we compute the modular commutators. We note that this construction only works well if $\beta_L,\beta_R \gg t^{-1}$. This is because the system is no longer described by a CFT at high energies. At finite sizes the same construction also works but one needs to substitute the integrals in $k$ by a finite sum over discrete momenta. The CFT computation in the main text implies the modular commutator 
	\begin{equation}
		J(A,B,C) = \frac{\pi}{3} (\eta^{\beta_L}_{\mathrm{eff}}(L_A,L_B,L_C)-\eta^{\beta_R}_{\mathrm{eff}}(L_A,L_B,L_C)),
	\end{equation}
	where $\eta^{\beta}_{\mathrm{eff}}$ is defined in Eq.~\eqref{eq:eta_eff_appendix}.
	As $\beta\rightarrow\infty$, $\eta^{\beta}_{\mathrm{eff}}\rightarrow (x_{12}x_{34})/(x_{13}x_{24}) =\eta$. In the limit in which all the intervals are significantly larger than $\beta$, $\eta^{\beta}_{\mathrm{eff}}\rightarrow 0$. Below we study two examples, both with $t=1/2$. (1) $(\beta_L,\beta_R) = (\infty,80)$, and (2) $(\beta_L,\beta_R)=(82,78)$. The first example is relevant to the edge of chiral topological order, and the second example is close to the time reversal invariant case so we can expect $J$ to be small. In the first case we also study two ranges of the subsystem sizes which are much larger or comparable to $\beta_R$. If all sizes are large compared to $\beta_R$ then we get $J= \pi c \eta/3$, like on the edge of chiral topological order. Our numerical result perfectly agrees with the predictions.
	\begin{figure}
		\centering
		\subfloat[$(\beta_L,\beta_R)=(\infty,80),L_A=L_C=320,120\leq L_B\leq 800$]{\includegraphics[width=0.32\linewidth]{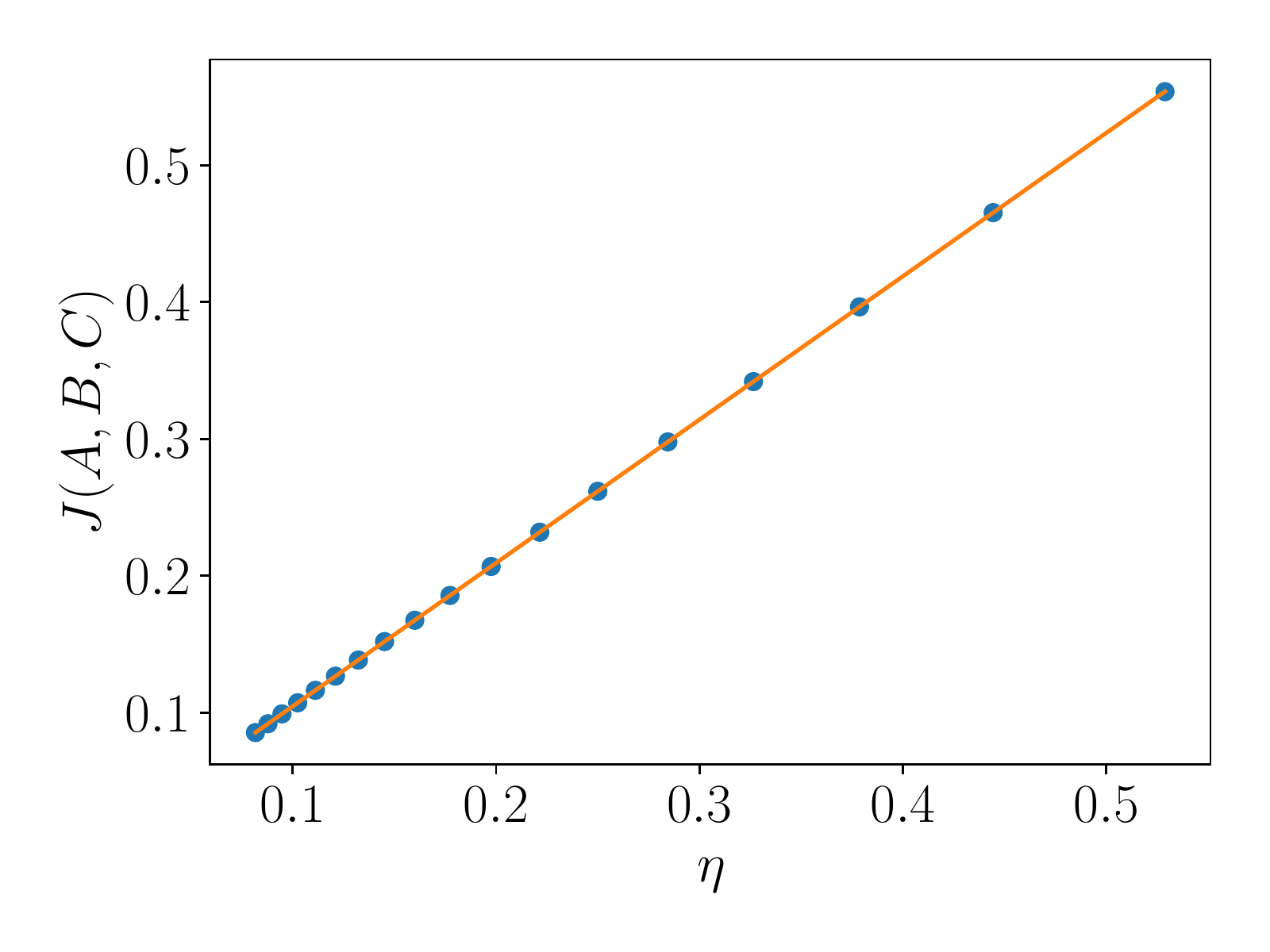}}
		\subfloat[$(\beta_L,\beta_R)=(\infty,80),L_A=L_C=40,10\leq L_B\leq 80$]{\includegraphics[width=0.32\linewidth]{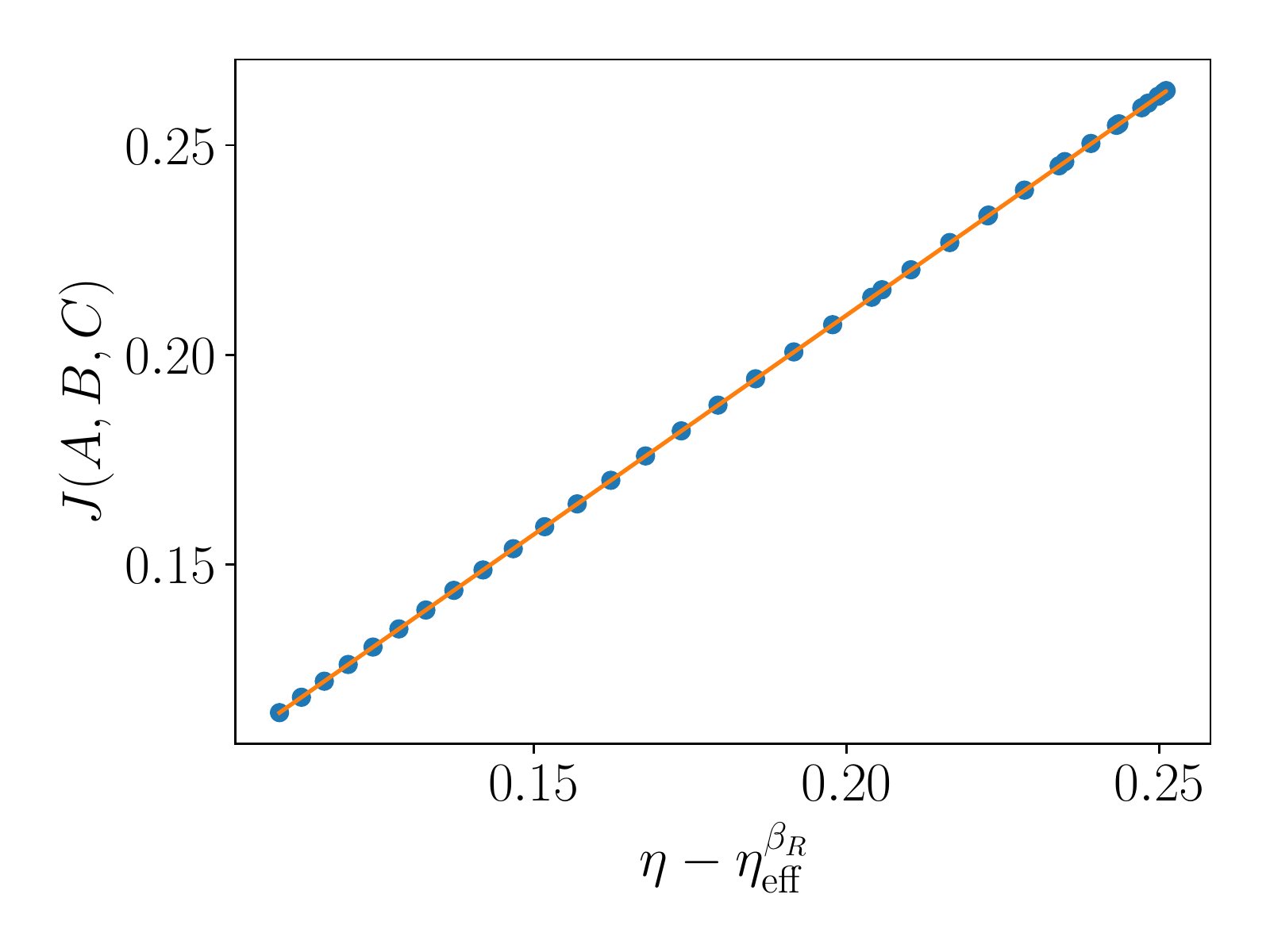}}
		\subfloat[$(\beta_L,\beta_R)=(82,78),L_A=L_C=40,10\leq L_B\leq 80$]{\includegraphics[width=0.32\linewidth]{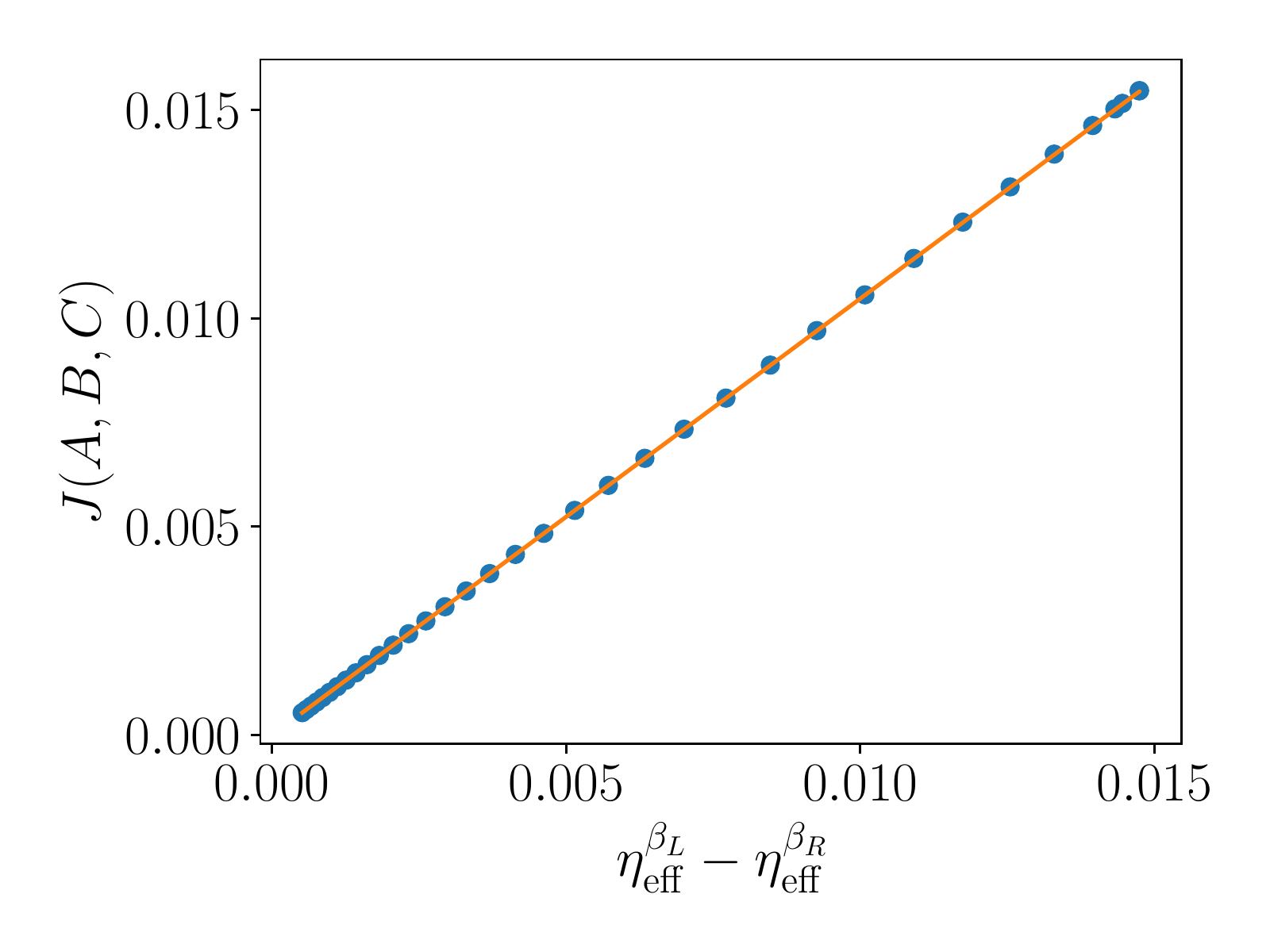}}
		\caption{Modular commutators in chiral thermal state of free fermion CFT on an infinite line.}
		\label{fig:chiral_thermal}
	\end{figure}
	
\end{document}